\documentclass[11pt,a4paper]{article}

\usepackage[utf8]{inputenc}
\usepackage[T1]{fontenc}
\usepackage{lmodern}

\usepackage{amsmath,amssymb,amsthm}
\usepackage{bm}
\usepackage{mathtools}

\usepackage{booktabs}
\usepackage{array}
\usepackage{float}          %

\usepackage{tikz}
\usepackage{pgfplots}
\pgfplotsset{compat=1.18}
\usetikzlibrary{arrows.meta, patterns, calc, decorations.markings}
\usepgfplotslibrary{fillbetween}

\usepackage[margin=2.5cm]{geometry}
\usepackage{setspace}
\usepackage{parskip}

\usepackage[authoryear,round]{natbib}

\usepackage[colorlinks=true,linkcolor=blue,citecolor=blue,urlcolor=blue]{hyperref}
\usepackage{cleveref}

\newtheorem{proposition}{Proposition}
\newtheorem{corollary}{Corollary}

\theoremstyle{definition}
\newtheorem{definition}{Definition}
\theoremstyle{remark}
\newtheorem*{remark}{Remark}

\newcommand{\dd}{\mathrm{d}}
\newcommand{\E}{\mathbb{E}}

\newcommand{\tw}{\tau_w}

\title{Wealth Taxation as a Drift Modification:\\
A Fokker--Planck Approach to Tax Neutrality}
\author{Anders G Fr{\o}seth\thanks{Independent Researcher.
  E-mail: \href{mailto:indrefjorden@pm.me}{indrefjorden@pm.me}.}}
\date{\today}

\begin{document}
\maketitle

\begin{abstract}
We reformulate the neutral wealth tax framework of \citet{Froeseth2026N}
in the language of stochastic dynamics and statistical physics.
Individual wealth under geometric Brownian motion satisfies a Langevin
equation with multiplicative noise; the probability density of wealth
across a population then evolves according to a Fokker--Planck equation.
A proportional wealth tax at market value enters as a uniform reduction
of the drift coefficient, preserving the diffusion structure and all
relative probability currents.  This drift-shift symmetry is the
physical content of tax neutrality.  Each channel through which
neutrality breaks down in practice---book-value assessment, liquidity
frictions, forced dividend extraction, migration, and market
impact---corresponds to a specific violation of this symmetry: a
state-dependent, asset-dependent, or flow-dependent modification of the
Fokker--Planck equation.  The framework clarifies when wealth taxation is
a benign rescaling of the dynamics and when it introduces genuinely new
physics.
\end{abstract}

\section{Introduction}\label{sec:intro}

The wealth tax debate, both in the academic literature and in public
policy, typically proceeds in one of two registers.  Economists analyse
the tax through equilibrium models---CAPM, Modigliani--Miller,
consumption-based asset pricing---where the object of interest is the
effect on prices, returns, and portfolio weights.  Separately, a large
econophysics literature studies wealth distributions using tools from
statistical mechanics---Boltzmann--Gibbs ensembles, Fokker--Planck
equations, agent-based models---but rarely connects to the asset-pricing
framework that governs how taxes interact with financial markets.

This paper bridges the two.  We show that the neutrality results derived
in \citet{Froeseth2026N} and the distortion channels analysed in
\citet{Froeseth2026E} have a natural and precise formulation in the
language of stochastic dynamics.  The mapping is not metaphorical: the
same geometric Brownian motion that underlies the finance results
\emph{is} a Langevin equation, and the evolution of the wealth
distribution across investors \emph{is} governed by a Fokker--Planck
equation.  No new physics is introduced; we simply read the existing
mathematics in a different dialect.

The value of this exercise is threefold.  First, the Fokker--Planck
formulation makes the symmetry content of neutrality transparent: the
wealth tax is neutral if and only if it enters as a uniform,
state-independent shift of the drift coefficient.  Second, each
violation of neutrality maps to a specific, classifiable modification of
the Fokker--Planck equation---offering a taxonomy of distortions that is
both exhaustive and physically intuitive.  Third, the framework
naturally accommodates questions about wealth \emph{distributions} and
their evolution---relaxation times, steady states, and the interplay
between individual dynamics and aggregate outcomes---that lie outside
the scope of representative-agent asset pricing.

We proceed step by step.  \Cref{sec:cochrane} introduces the core
concepts of modern asset pricing---the stochastic discount factor,
no-arbitrage, and the price of risk---in a language accessible to
physicists.  \Cref{sec:langevin} establishes the core mapping:
individual wealth dynamics under geometric Brownian motion as a
Langevin equation with multiplicative noise.  \Cref{sec:fp} derives the
corresponding Fokker--Planck equation for the wealth distribution.
\Cref{sec:tax} introduces the proportional wealth tax and shows that it
enters as a pure drift shift.  \Cref{sec:neutrality} formalises the
neutrality result as a symmetry of the Fokker--Planck equation and
shows that this symmetry is robust to non-Gaussian returns and
stochastic volatility.
\Cref{sec:channels} maps each distortion channel to a specific
symmetry-breaking mechanism.  \Cref{sec:steady} discusses steady-state
distributions when income and consumption are included as source and
sink terms.  \Cref{sec:discussion} discusses extensions and open
questions.

\section{Asset pricing for physicists: Cochrane's framework}\label{sec:cochrane}

This section introduces the core ideas of modern asset pricing in a
language accessible to readers trained in physics or physical chemistry.
The presentation follows \citet{Cochrane2005}, who showed that the
entire theory reduces to a single equation and a single object---the
\emph{stochastic discount factor}.  Readers familiar with finance may
skip to \Cref{sec:langevin}.

\subsection{What is an asset?}

An \emph{asset} is anything that produces an uncertain future payoff.
A share of stock pays dividends and can be sold; a bond pays coupons and
returns its face value; a house provides rent (or imputed rent) and can
be resold.  The central question of asset pricing is: \emph{given the
uncertain future payoff, what is the correct price today?}

In physics terms, an asset is a random variable~$x_{t+1}$ representing
the total payoff (cash flow plus resale value) received at time
$t+1$.  The price~$p_t$ is the value assigned to this random variable
today.  The \emph{gross return} is
\begin{equation}\label{eq:return_def}
  R_{t+1} \equiv \frac{x_{t+1}}{p_t} \,,
\end{equation}
so that $R_{t+1} = 1.05$ means a $5\%$ gain.  The return is a
dimensionless ratio---the payoff per unit of price---and is the
financial analogue of a growth factor.

\subsection{The fundamental equation of asset pricing}

Cochrane's central result is that, under very weak assumptions (no
arbitrage, discussed below), there exists a random variable
$m_{t,t+1}$---the \emph{stochastic discount factor} (SDF)---such that
the price of \emph{every} asset satisfies
\begin{equation}\label{eq:sdf}
  \boxed{p_t = \E_t\bigl[m_{t,t+1} \; x_{t+1}\bigr] \,.}
\end{equation}
Here $\E_t[\cdot]$ denotes the conditional expectation given all
information available at time~$t$.  The equation says: the price today
equals the expected value of tomorrow's payoff, weighted by~$m$.

\begin{remark}[Analogy with statistical mechanics]
Equation~\eqref{eq:sdf} has the same mathematical structure as a
\emph{partition-function average}.  In statistical mechanics, the
expectation value of an observable~$A$ in a canonical ensemble is
\begin{equation}\label{eq:stat_mech_avg}
  \langle A \rangle = \sum_s A(s) \; \frac{e^{-\beta E(s)}}{Z} \,,
\end{equation}
where $\beta = 1/k_B T$ is the inverse temperature and
$Z = \sum_s e^{-\beta E(s)}$ is the partition function.  The Boltzmann
weight $e^{-\beta E(s)}/Z$ plays exactly the role of~$m$: it assigns a
positive weight to each state, and the observable's ``price'' (its
equilibrium expectation) is the weighted average over states.

The analogy is not perfect---$m$ is stochastic and need not take the
exponential Boltzmann form---but the structural parallel is deep:
both frameworks assign state-dependent weights to compute expectations.
In finance, ``bad states'' (recessions, crises) receive high weight
because investors value payoffs more when they are poor; in physics,
low-energy states receive high weight because they are thermodynamically
favoured.
\end{remark}

\subsection{The stochastic discount factor}

The SDF~$m_{t,t+1}$ encodes how the market values payoffs in different
states of the world.  Its key properties:

\textbf{Positivity.}  Under no-arbitrage (see below), $m > 0$ in all
states.  This is the financial analogue of the requirement that
Boltzmann weights are positive.

\textbf{High in bad states.}  When the economy is in a bad state
(low consumption, high unemployment), $m$ is large: the market assigns
high value to payoffs received precisely when they are most needed.
This is the origin of risk premia.

\textbf{Low in good states.}  Conversely, payoffs received when
everyone is already wealthy are worth less per unit.

In consumption-based models, $m$ has the explicit form
\begin{equation}\label{eq:m_consumption}
  m_{t,t+1} = \beta \left(\frac{C_{t+1}}{C_t}\right)^{-\gamma} ,
\end{equation}
where $\beta < 1$ is a time-preference (patience) parameter,
$C_t$ is aggregate consumption, and $\gamma > 0$ is relative risk
aversion.  When consumption falls ($C_{t+1} < C_t$), the ratio
$(C_{t+1}/C_t)^{-\gamma}$ is large, so $m$ is large: payoffs in
recessions are highly valued.  But Cochrane's key insight is that
\eqref{eq:sdf} holds \emph{regardless} of whether we assume
\eqref{eq:m_consumption} or any other specific model.  The SDF exists
as long as there is no arbitrage.

\subsection{No arbitrage: the financial conservation law}

\emph{No arbitrage} means that there is no trading strategy that
produces a positive expected payoff with zero cost and zero risk.  In
physics language: \emph{there is no perpetual motion machine in
financial markets.}

\begin{proposition}[Fundamental theorem of asset pricing]
The following are equivalent:
\begin{enumerate}
  \item There are no arbitrage opportunities.
  \item There exists a strictly positive stochastic discount factor
    $m > 0$ such that $p = \E[m \, x]$ for all traded assets.
\end{enumerate}
\end{proposition}

This is the most important theorem in finance.  It says that the
absence of free lunches is \emph{equivalent} to the existence of a
positive weighting function~$m$.  The parallel to physics is:
\begin{center}
\renewcommand{\arraystretch}{1.4}
\begin{tabular}{@{}ll@{}}
\toprule
\textbf{Physics} & \textbf{Finance} \\
\midrule
No perpetual motion (2nd law) & No arbitrage \\
$\Rightarrow$ Entropy is a state function & $\Rightarrow$ SDF exists ($m > 0$) \\
$\Rightarrow$ Boltzmann weights well-defined & $\Rightarrow$ Prices are expectations under $m$ \\
\bottomrule
\end{tabular}
\end{center}
Just as the second law of thermodynamics is not derived from
microscopic dynamics but constrains all possible dynamics, no-arbitrage
is not derived from individual behaviour but constrains all possible
prices.

\subsection{Risk, return, and the price of risk}

Dividing \eqref{eq:sdf} by $p_t$ and using $R_{t+1} = x_{t+1}/p_t$
gives the \emph{return form}:
\begin{equation}\label{eq:sdf_return}
  1 = \E_t[m_{t,t+1} \, R_{t+1}] \,.
\end{equation}
Expanding using the covariance identity
$\E[XY] = \E[X]\E[Y] + \mathrm{Cov}(X,Y)$:
\begin{equation}\label{eq:risk_premium}
  \E_t[R_{t+1}] = \frac{1}{\E_t[m_{t,t+1}]}
    - \frac{\mathrm{Cov}_t(m_{t,t+1}, R_{t+1})}{\E_t[m_{t,t+1}]} \,.
\end{equation}
The first term is the \emph{risk-free rate}: the return on an asset
whose payoff does not fluctuate.  Since a risk-free asset has
$\mathrm{Cov}(m, R_f) = 0$:
\begin{equation}\label{eq:rf}
  R_f = \frac{1}{\E[m]} \,.
\end{equation}
The second term is the \emph{risk premium}: assets that pay off in bad
states (when $m$ is high) have $\mathrm{Cov}(m, R) > 0$, hence
\emph{lower} expected returns---they provide insurance.  Assets that
pay off in good states (when $m$ is low) have
$\mathrm{Cov}(m, R) < 0$, hence \emph{higher} expected returns---they
carry risk.

\begin{remark}[Risk is not variance]
A central insight that physicists often find surprising: in finance,
risk is \emph{not} measured by the variance of an asset's return.  An
asset can be highly volatile yet command no risk premium if its
fluctuations are uncorrelated with the SDF (i.e., uncorrelated with the
aggregate state of the economy).  Risk is entirely about
\emph{covariance with the pricing kernel}---not individual variance.
This is analogous to the fact that, in statistical mechanics, the
contribution of a mode to the free energy depends on its coupling to
the thermal bath, not on its amplitude alone.
\end{remark}

\subsection{The risk-free rate as temperature}

The risk-free rate $R_f = 1/\E[m]$ sets the baseline price of time:
how much the market discounts future payoffs purely for waiting, with
no risk.  It plays a role analogous to temperature in statistical
mechanics:
\begin{itemize}
  \item \textbf{Low $R_f$} (low interest rates) $\Leftrightarrow$
    $\E[m]$ is high $\Leftrightarrow$ future payoffs are valuable
    $\Leftrightarrow$ asset prices are high.  This is a ``hot'' market
    in physics language: high prices, low discount rates, willingness to
    pay for future outcomes.
  \item \textbf{High $R_f$} (high interest rates) $\Leftrightarrow$
    $\E[m]$ is low $\Leftrightarrow$ future payoffs are less valuable
    $\Leftrightarrow$ asset prices are low.  A ``cold'' market:
    investors demand compensation for waiting.
\end{itemize}
The analogy is suggestive but should be used with care: unlike
thermodynamic temperature, $R_f$ is not a state variable of the
system---it is an endogenous price determined by aggregate patience
and growth expectations.

\subsection{Specific models as specifications of $m$}

Different finance models are simply different choices for the form
of~$m$:

\begin{center}
\renewcommand{\arraystretch}{1.4}
\begin{tabular}{@{}lll@{}}
\toprule
\textbf{Model} & \textbf{SDF specification} & \textbf{Physics analogue} \\
\midrule
CAPM & $m = a - b \, R_{\text{market}}$
  & Single-mode coupling \\
Consumption CAPM & $m = \beta(C_{t+1}/C_t)^{-\gamma}$
  & Boltzmann weight \\
Fama--French 3-factor & $m = a - b_1 f_1 - b_2 f_2 - b_3 f_3$
  & Three-mode coupling \\
Black--Scholes & $m = e^{-r_f \Delta t}$ (risk-neutral)
  & Constant weight \\
\bottomrule
\end{tabular}
\end{center}

The power of Cochrane's approach is that \eqref{eq:sdf} holds
\emph{before} choosing a specific model.  This is analogous to
deriving results from thermodynamic identities (which hold for any
system) rather than from a specific Hamiltonian.

\subsection{Cochrane's framework and the wealth tax}

With this machinery, we can state the wealth tax neutrality result in
Cochrane's language.  A proportional wealth tax at rate~$\tw$ reduces
every asset's payoff by the same factor: $x_{t+1} \to (1-\tw)\, x_{t+1}$.
Since~$m$ is unchanged (the tax does not alter the state of the
economy or the marginal utility of consumption), the price becomes
\begin{equation}\label{eq:sdf_tax}
  p_t^{\text{after tax}} = \E_t\bigl[m_{t,t+1}\;(1-\tw)\,x_{t+1}\bigr]
  = (1-\tw)\,\E_t[m_{t,t+1}\;x_{t+1}]
  = (1-\tw)\,p_t \,.
\end{equation}
Every price falls by the same factor $(1-\tw)$.  The \emph{return}
$R_{t+1} = x_{t+1}/p_t$ is unchanged because both numerator and
denominator scale by $(1-\tw)$.  Expected returns, risk premia,
Sharpe ratios, and optimal portfolio weights are all invariant.

This is the pricing side of the neutrality result.  The remainder of
the paper shows that the same result has a natural expression in
Fokker--Planck language: the tax is a uniform drift shift that
preserves the structure of the stochastic dynamics.

\begin{remark}[When does the SDF change?]
The neutrality result assumes that $m$ is unaffected by the tax.
This holds in partial equilibrium (the tax does not change aggregate
consumption or the marginal utility structure) and in the general
equilibrium of a homogeneous economy where all investors are taxed
identically.  If the tax is non-uniform---applying only to some
investors or some assets---then $m$ may change, and neutrality breaks
down.  This connects to the distortion channels of
\Cref{sec:channels}.
\end{remark}

\subsection{Correspondence of concepts}\label{sec:rosetta}

\Cref{tab:rosetta} collects the key conceptual correspondences between
the finance and statistical physics frameworks developed in this paper.
The table is intended as a reference; each entry is introduced and
justified in the sections that follow.

\begin{table}[H]
\centering
\caption{Correspondence between finance and statistical physics
  concepts.  The left column lists standard finance terminology
  (following \citealt{Cochrane2005}); the right column gives the
  statistical physics equivalent used in this paper.  Entries above the
  mid-rule are conceptual; entries below are mathematical objects.}
\label{tab:rosetta}
\renewcommand{\arraystretch}{1.45}
\small
\begin{tabular}{@{}p{0.42\textwidth}p{0.42\textwidth}@{}}
\toprule
\textbf{Finance} & \textbf{Statistical physics} \\
\midrule
\multicolumn{2}{@{}l}{\textit{Entities and structure}} \\[3pt]
Investor & Particle (trajectory in wealth space) \\
Population of investors & Ensemble \\
Asset (stock, bond, \ldots) & Degree of freedom \\
Portfolio of $N$ assets & System of $N$ coupled degrees of freedom \\
Market & Open system (with sources and sinks) \\[6pt]
\multicolumn{2}{@{}l}{\textit{Prices and constraints}} \\[3pt]
No arbitrage & No perpetual motion (2nd law) \\
Stochastic discount factor $m$ & Boltzmann weight / measure on states \\
Asset price $p = \E[m\,x]$ & Partition-function average $\langle A \rangle$ \\
Risk-free rate $R_f = 1/\E[m]$ & Inverse temperature $\beta = 1/k_BT$
  \emph{(suggestive, not exact)} \\[6pt]
\multicolumn{2}{@{}l}{\textit{Dynamics}} \\[3pt]
Expected return $\mu$ & Drift velocity $v = \mu - \sigma^2/2$ \\
Volatility $\sigma$ & Noise strength (diffusion coeff.${}^{1/2}$) \\
GBM $\dd W/W = \mu\,\dd t + \sigma\,\dd B_t$
  & Langevin equation (multiplicative noise) \\
Log-return process
  & Langevin equation (additive noise) \\
Diffusion coefficient & $D = \sigma^2/2$ (Einstein relation) \\
Covariance matrix $\bm{\Sigma}$ & Coupling matrix between modes \\
Eigenvectors of $\bm{\Sigma}$ & Normal modes of the system \\[6pt]
\multicolumn{2}{@{}l}{\textit{Distributions and equilibrium}} \\[3pt]
Wealth distribution $p(W,t)$ & Probability density (FP equation) \\
Probability current $J$ & Net flux in wealth space \\
Efficient market (risk-neutral drift)
  & Detailed balance under pricing measure \\
Wealth inequality (Pareto tail)
  & Power-law tail of steady-state distribution \\
Gini coefficient, Pareto exponent
  & Moments / exponents of $\pi_{\mathrm{ss}}(x)$ \\[6pt]
\multicolumn{2}{@{}l}{\textit{Taxation}} \\[3pt]
Proportional wealth tax $\tw$
  & Uniform drift shift: $v \to v - \tw$ \\
Tax neutrality & Drift-shift symmetry (uniform external field) \\
Book-value assessment & Anisotropic field (asset-dependent coupling) \\
Liquidity friction & State-dependent friction ($v, D$ depend on $W$) \\
Forced dividend extraction & Coupling to slow variable (firm capital $K$) \\
Migration (tax-induced emigration) & Absorbing boundary / sink term \\
Market impact of forced sales & Mean-field interaction (nonlinear drift) \\
\bottomrule
\end{tabular}
\end{table}

\section{Individual wealth as a Langevin equation}\label{sec:langevin}

\subsection{Geometric Brownian motion in financial language}

Consider an investor whose entire wealth $W(t)$ is invested in a single
risky asset.  Under the standard assumption of geometric Brownian motion
(GBM), the wealth evolves as
\begin{equation}\label{eq:gbm}
  \frac{\dd W}{W} = \mu \, \dd t + \sigma \, \dd B_t \,,
\end{equation}
where $\mu$ is the expected instantaneous return, $\sigma > 0$ is the
volatility, and $B_t$ is a standard Brownian motion.  This is the
starting point of the Black--Scholes--Merton framework, the Capital
Asset Pricing Model, and the neutrality analysis in
\citet{Froeseth2026N}.

The key property of \eqref{eq:gbm} is that the noise is
\emph{multiplicative}: the magnitude of the random shock $\sigma W \,
\dd B_t$ is proportional to the current wealth level.  Rich investors
experience larger absolute fluctuations than poor investors, even though
the percentage fluctuation~$\sigma$ is the same.

\subsection{Log-wealth and the Langevin equation}

Define $x(t) \equiv \ln W(t)$.  By It\^{o}'s lemma,
\begin{equation}\label{eq:langevin}
  \dd x = \underbrace{\left(\mu - \tfrac{\sigma^2}{2}\right)}_{\displaystyle
    \equiv\, v} \dd t \;+\; \sigma \, \dd B_t \,.
\end{equation}
This is a \emph{Langevin equation with additive noise} in the
log-wealth coordinate~$x$.  The drift velocity $v = \mu - \sigma^2/2$
is constant, and the noise strength~$\sigma$ is independent of the
state.  In the physics of Brownian motion, \eqref{eq:langevin}
describes a particle drifting at constant velocity~$v$ in a viscous
medium subject to thermal fluctuations of strength~$\sigma$
\citep[Ch.~2]{Zwanzig2001}.

\begin{remark}[No restoring force]
There is no spring constant, no mean-reversion, and no equilibrium
position.  The particle drifts indefinitely.  This distinguishes the
mapping from a harmonic oscillator, and the distinction matters: a
harmonic oscillator has a stationary Boltzmann distribution, while a
freely drifting Brownian particle does not.  For the wealth distribution
to have a steady state, additional ingredients (income, consumption,
death) are required; see \Cref{sec:steady}.
\end{remark}

\subsection{Summary of the mapping}

The key identifications---wealth as $e^x$, expected return as drift
velocity, volatility as noise strength, $D = \sigma^2/2$ as the Einstein
relation---are collected in the ``Dynamics'' rows of
\Cref{tab:rosetta}.  No approximation is involved: the two columns are
the same mathematical object read in different notation.

\subsection{Multiple assets and the portfolio}

When the investor holds $N$ assets with return vector
$\bm{\mu} = (\mu_1, \ldots, \mu_N)^\top$, volatilities
$\bm{\sigma} = (\sigma_1, \ldots, \sigma_N)^\top$, and correlation
matrix $\bm{\rho}$, the portfolio return is
\begin{equation}\label{eq:portfolio_return}
  \frac{\dd W}{W} = \mathbf{w}^\top \bm{\mu} \, \dd t
    + \mathbf{w}^\top \bm{\Sigma}^{1/2} \, \dd \mathbf{B}_t \,,
\end{equation}
where $\mathbf{w}$ is the vector of portfolio weights and
$\bm{\Sigma} = \mathrm{diag}(\bm{\sigma})\, \bm{\rho} \,
\mathrm{diag}(\bm{\sigma})$ is the covariance matrix.  The
log-wealth again follows a Langevin equation:
\begin{equation}\label{eq:langevin_port}
  \dd x = \left(\mathbf{w}^\top \bm{\mu}
    - \tfrac{1}{2}\mathbf{w}^\top \bm{\Sigma} \, \mathbf{w}\right)
    \dd t
    \;+\; \sqrt{\mathbf{w}^\top \bm{\Sigma}\, \mathbf{w}} \;\dd B_t \,,
\end{equation}
where the scalar Brownian motion~$B_t$ drives the portfolio return
and the effective diffusion coefficient is
$D_{\mathrm{port}} = \tfrac{1}{2}\mathbf{w}^\top \bm{\Sigma}\,
\mathbf{w}$.

\begin{remark}[Correlation structure]
The covariance matrix~$\bm{\Sigma}$ plays the role of a coupling
matrix in the multi-dimensional Langevin system.  Its eigenvectors
define the independent ``modes'' of the system: the largest eigenvalue
corresponds to the aggregate market factor; smaller eigenvalues
correspond to sector or idiosyncratic modes.  This is mathematically
identical to the normal-mode decomposition of coupled linear systems in
physics, though the dynamics here are stochastic rather than
deterministic.  Random matrix theory (the Mar\v{c}enko--Pastur law)
provides the tools to separate signal from noise in the empirical
covariance matrix---a problem with no counterpart in deterministic
coupled oscillators.
\end{remark}

\section{The Fokker--Planck equation for wealth}\label{sec:fp}

\subsection{From individual trajectories to distributions}

\begin{remark}[Terminology]
Throughout this paper, we use \emph{investor} for the individual
decision-maker in the financial market, consistent with the asset
pricing literature.  In the statistical physics mapping, each
investor's wealth trajectory corresponds to a \emph{particle}
trajectory in the Langevin/Fokker--Planck framework, and the population
of investors corresponds to the \emph{ensemble}.  We reserve the term
\emph{agent} for the distinct context of agent-based computational
models, which are not our primary framework here.
\end{remark}

Suppose there are $\mathcal{N}$ investors, each with wealth $W_i(t)$
following \eqref{eq:gbm} with common parameters $(\mu, \sigma)$.
Here $W$ denotes \emph{market net wealth}: the market value of all
assets minus the face value of all liabilities.  When the wealth tax
is levied at market value (the neutral case), the tax base equals~$W$
and the analysis is clean.  When assets are assessed below market
value while liabilities remain fully deductible, the tax base
diverges from~$W$---a distinction that becomes important in the
distortion analysis of \Cref{sec:ch1}.

Let $p(W, t)$ denote the probability density: $p(W,t)\, \dd W$ is the
fraction of investors with wealth in $[W, W + \dd W]$ at time~$t$.

The standard result from stochastic calculus (the forward Kolmogorov
equation) gives the Fokker--Planck equation corresponding to the
stochastic differential equation~\eqref{eq:gbm}:
\begin{equation}\label{eq:fp_W}
  \boxed{
  \frac{\partial p}{\partial t}
  = -\frac{\partial}{\partial W}\bigl[\mu W \, p\bigr]
    + \frac{1}{2}\frac{\partial^2}{\partial W^2}
      \bigl[\sigma^2 W^2 \, p\bigr] \,.}
\end{equation}
The first term is the \emph{drift} (or advection): wealth is pushed
upward at rate~$\mu W$.  The second term is the
\emph{diffusion}: the distribution spreads due to the stochastic
fluctuations of magnitude~$\sigma W$.

\subsection{Fokker--Planck in log-wealth coordinates}

Because the noise in \eqref{eq:gbm} is multiplicative, the
coefficients in \eqref{eq:fp_W} depend on~$W$.  The coordinate change
$x = \ln W$ removes this state-dependence.  Let $\pi(x,t)$ denote the
density of log-wealth, related to $p(W,t)$ by
\begin{equation}
  \pi(x,t) = W \, p(W,t) = e^x \, p(e^x, t) \,.
\end{equation}
Substituting into \eqref{eq:fp_W}, we obtain the Fokker--Planck
equation for $\pi(x,t)$:
\begin{equation}\label{eq:fp_x}
  \boxed{
  \frac{\partial \pi}{\partial t}
  = -v \frac{\partial \pi}{\partial x}
    + D \frac{\partial^2 \pi}{\partial x^2} \,,}
\end{equation}
where $v = \mu - \sigma^2/2$ is the drift velocity and
$D = \sigma^2/2$ is the diffusion coefficient.  This is the standard
\emph{drift--diffusion equation} with constant coefficients---the
simplest non-trivial Fokker--Planck equation, and one of the most
studied objects in statistical physics
\citep[Ch.~3]{Zwanzig2001}\citep[Ch.~6]{LiviPoliti2017}.

\begin{remark}[Propagator]
The Green's function of \eqref{eq:fp_x} is the Gaussian
\begin{equation}\label{eq:propagator}
  \pi(x, t \,|\, x_0, 0) = \frac{1}{\sqrt{4\pi D t}}
    \exp\!\left(-\frac{(x - x_0 - vt)^2}{4Dt}\right) ,
\end{equation}
which is simply the statement that log-wealth is normally distributed
with mean $x_0 + vt$ and variance $2Dt = \sigma^2 t$.  This is the
standard result that $W(t)$ is log-normally distributed under GBM.
(The Gaussian form of the propagator is specific to the
constant-coefficient case; the drift-shift symmetry that underlies
our neutrality result does not depend on it---see
\Cref{sec:robustness}.)  The
Fokker--Planck formulation adds nothing new for a single investor, but
becomes essential when we consider distributions across investors and
their evolution under policy changes.
\end{remark}

\subsection{Probability current}

The Fokker--Planck equation can be written as a continuity equation:
\begin{equation}\label{eq:continuity}
  \frac{\partial \pi}{\partial t} + \frac{\partial J}{\partial x} = 0 \,,
\end{equation}
where the \emph{probability current} is
\begin{equation}\label{eq:current}
  J(x, t) = v\, \pi(x,t) - D \frac{\partial \pi}{\partial x} \,.
\end{equation}
The current~$J$ has a direct interpretation: it is the net flux of
investors (in probability terms) flowing past the point~$x$ per unit
time.  The drift term $v\pi$ pushes probability to the right (higher
wealth); the diffusion term $-D\, \partial\pi/\partial x$ drives
probability from high-density to low-density regions
\citep[see][Ch.~6, for boundary conditions and current analysis]{LiviPoliti2017}.

The concept of probability current will be central to the neutrality
analysis: a tax is neutral if it does not alter the \emph{relative}
currents between any two points in wealth space.

\section{Proportional wealth tax as drift modification}\label{sec:tax}

\subsection{The taxed dynamics}

Now impose a proportional wealth tax at rate $\tw$ on the market value
of all holdings.  Following \citet{Froeseth2026N}, the after-tax
wealth dynamics are
\begin{equation}\label{eq:gbm_tax}
  \frac{\dd W}{W} = (\mu - \tw) \, \dd t + \sigma \, \dd B_t \,.
\end{equation}
The tax reduces the expected return from $\mu$ to $\mu - \tw$, but
leaves the volatility~$\sigma$ unchanged.  This is the multiplicative
separability result: the tax operates as a proportional dilution,
removing a fraction~$\tw$ of wealth per unit time, without altering
the stochastic structure of returns.

In log-wealth:
\begin{equation}\label{eq:langevin_tax}
  \dd x = \underbrace{\left(\mu - \tw - \tfrac{\sigma^2}{2}\right)}_{
    \displaystyle \equiv\, v_\tau} \dd t
    \;+\; \sigma \, \dd B_t \,.
\end{equation}
The only change is the drift velocity: $v \to v_\tau = v - \tw$.

\subsection{The taxed Fokker--Planck equation}

The Fokker--Planck equation for the taxed system is
\begin{equation}\label{eq:fp_tax}
  \boxed{
  \frac{\partial \pi}{\partial t}
  = -v_\tau \frac{\partial \pi}{\partial x}
    + D \frac{\partial^2 \pi}{\partial x^2} \,,}
\end{equation}
with $v_\tau = v - \tw$ and $D = \sigma^2/2$ unchanged.  The taxed
probability current is
\begin{equation}\label{eq:current_tax}
  J_\tau(x, t) = v_\tau\, \pi(x,t)
    - D \frac{\partial \pi}{\partial x} \,.
\end{equation}

\subsection{What the tax does and does not change}

\begin{center}
\renewcommand{\arraystretch}{1.4}
\begin{tabular}{@{}lcc@{}}
\toprule
\textbf{Quantity} & \textbf{Untaxed} & \textbf{Taxed} \\
\midrule
Drift velocity & $v = \mu - \sigma^2/2$ & $v_\tau = \mu - \tw - \sigma^2/2$ \\
Diffusion coefficient & $D = \sigma^2/2$ & $D = \sigma^2/2$ \\
Noise strength & $\sigma$ & $\sigma$ \\
Propagator width & $\sqrt{2Dt}$ & $\sqrt{2Dt}$ \\
Propagator centre & $x_0 + vt$ & $x_0 + v_\tau t$ \\
\bottomrule
\end{tabular}
\end{center}
\smallskip

The tax translates the centre of the propagator to the left (lower
average wealth) but does not change its width (same dispersion around
the mean).  In physical terms: the particle drifts more slowly, but
diffuses at the same rate.  \Cref{fig:propagator} illustrates this
for a population of investors evolving over twenty years.

\begin{figure}[H]
\centering
\begin{tikzpicture}
\begin{axis}[
  width=0.85\textwidth,
  height=0.48\textwidth,
  xlabel={Log-wealth $x = \ln W$},
  ylabel={Probability density $\pi(x,t)$},
  domain=-2:8,
  samples=200,
  ymin=0,
  ymax=0.58,
  xmin=-2,
  xmax=8,
  legend style={
    at={(0.03,0.97)},
    anchor=north west,
    font=\small,
    draw=black!50,
    fill=white,
    fill opacity=0.9,
  },
  every axis plot/.append style={thick},
  axis lines=left,
  clip=true,
  tick label style={font=\small},
  label style={font=\small},
]

\addplot[black, dashed, thin, name path=initial]
  {1/(sqrt(2*3.14159*0.35))*exp(-(x-3)^2/(2*0.35))};
\addlegendentry{Initial ($t=0$)}

\addplot[blue!80!black, solid, name path=untaxed]
  {1/(sqrt(2*3.14159*0.8))*exp(-(x-4.6)^2/(2*0.8))};
\addlegendentry{Untaxed ($t=20$)}

\addplot[red!80!black, solid, name path=taxed]
  {1/(sqrt(2*3.14159*0.8))*exp(-(x-4.0)^2/(2*0.8))};
\addlegendentry{Taxed ($t=20$, $\tau_w = 3\%$)}

\draw[blue!80!black, densely dotted, thin]
  (axis cs:4.6,0) -- (axis cs:4.6,0.455);
\draw[red!80!black, densely dotted, thin]
  (axis cs:4.0,0) -- (axis cs:4.0,0.455);

\draw[-{Stealth[length=3mm]}, thick, black!60]
  (axis cs:4.6,0.50) -- (axis cs:4.0,0.50)
  node[midway, above, font=\small] {$\tau_w t$};

\draw[blue!80!black, densely dotted, very thin, opacity=0.5]
  (axis cs:3.706,0) -- (axis cs:3.706,0.275);
\draw[blue!80!black, densely dotted, very thin, opacity=0.5]
  (axis cs:5.494,0) -- (axis cs:5.494,0.275);
\draw[red!80!black, densely dotted, very thin, opacity=0.5]
  (axis cs:3.106,0) -- (axis cs:3.106,0.275);
\draw[red!80!black, densely dotted, very thin, opacity=0.5]
  (axis cs:4.894,0) -- (axis cs:4.894,0.275);

\draw[{Stealth[length=2mm]}-{Stealth[length=2mm]}, thick, blue!60!black, opacity=0.6]
  (axis cs:3.706,0.015) -- (axis cs:5.494,0.015);
\node[font=\footnotesize, blue!60!black, opacity=0.8, fill=white, inner sep=1pt]
  at (axis cs:4.6,0.015) {$\sqrt{2Dt}$};

\draw[{Stealth[length=2mm]}-{Stealth[length=2mm]}, thick, red!60!black, opacity=0.6]
  (axis cs:3.106,0.04) -- (axis cs:4.894,0.04);
\node[font=\footnotesize, red!60!black, opacity=0.8, fill=white, inner sep=1pt]
  at (axis cs:4.0,0.04) {$\sqrt{2Dt}$};

\end{axis}
\end{tikzpicture}
\caption{The separability result visualised.  A population of
  investors starts at log-wealth $x_0 = 3$ (dashed).  After $t = 20$
  years with $\mu = 10\%$ and $\sigma = 20\%$, the untaxed
  distribution (blue) and the distribution under a $\tau_w = 3\%$
  proportional wealth tax (red) have identical shape and width
  ($\sqrt{2Dt} = 0.89$) but different centres.  The tax shifts the
  propagator to the left by $\tau_w t = 0.6$ in log-wealth without
  deforming it.  This is the visual content of neutrality: the tax is
  a pure translation in log-wealth space.}
\label{fig:propagator}
\end{figure}
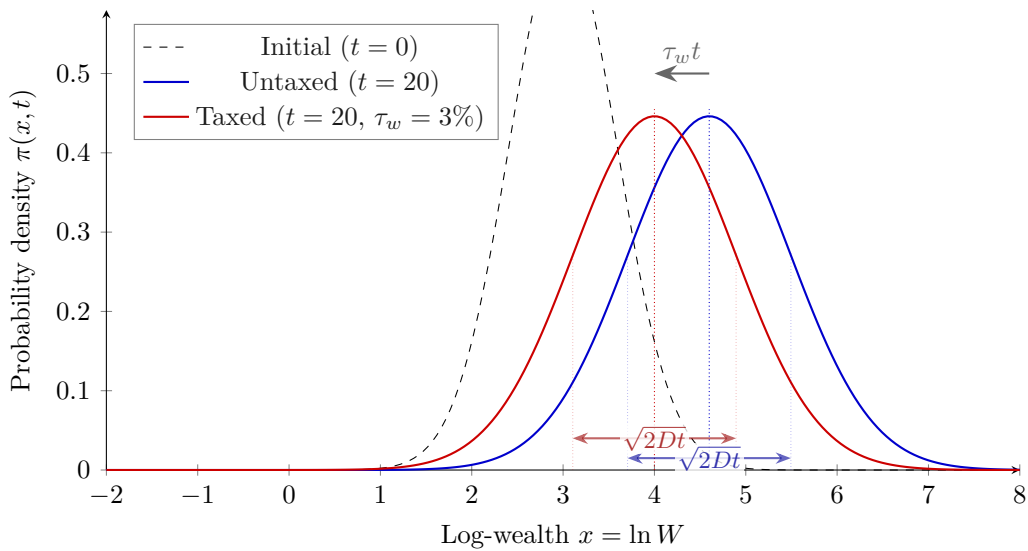

\section{Neutrality as a symmetry}\label{sec:neutrality}

\subsection{The drift-shift symmetry}

The central observation is that the proportional wealth tax acts as a
\emph{Galilean-type boost} in log-wealth space: it shifts the drift
velocity uniformly, without coupling to the state, the diffusion, or
any other parameter.

\begin{definition}[Drift-shift transformation]
  For $\tw \geq 0$, define the map
  $\mathcal{T}_\tau: v \mapsto v - \tw$, $D \mapsto D$.
  The taxed Fokker--Planck operator is
  $\mathcal{L}_\tau = \mathcal{T}_\tau \circ \mathcal{L}_0$.
\end{definition}

\begin{proposition}[Neutrality as invariance]\label{prop:neutrality}
Let two assets have drift--diffusion parameters
$(v_1, D_1)$ and $(v_2, D_2)$ under the untaxed Fokker--Planck
equation.  Under the proportional wealth tax at rate~$\tw$, the
parameters become $(v_1 - \tw, D_1)$ and $(v_2 - \tw, D_2)$.  Then:
\begin{enumerate}
  \item The difference in drift velocities is unchanged:
    $(v_1 - \tw) - (v_2 - \tw) = v_1 - v_2$.
  \item The ratio of diffusion coefficients is unchanged:
    $D_1 / D_2$.
  \item The Sharpe-ratio-like quantity
    $(v_i - v_j) / \sqrt{D_i + D_j - 2D_{ij}}$
    is unchanged for all pairs $(i,j)$.
  \item The optimal portfolio weights, which depend only on
    $\bm{\mu} - r_f \mathbf{1}$ and $\bm{\Sigma}$, are unchanged
    because $(\bm{\mu} - \tw\mathbf{1}) - (r_f - \tw)\mathbf{1}
    = \bm{\mu} - r_f\mathbf{1}$.
\end{enumerate}
\end{proposition}

\begin{proof}
Properties (i)--(iii) follow directly from the linearity of
$\mathcal{T}_\tau$ and the fact that it shifts all drifts by the same
constant.  Property (iv) uses the fact that the risk-free rate is also
subject to the wealth tax: $r_f \to r_f - \tw$.  The excess return
$\mu_i - r_f$ is invariant.  Since the Markowitz optimisation depends
only on excess returns and the covariance matrix~$\bm{\Sigma}$
(which involves only the diffusion coefficients), the optimal weights
are unchanged.
\end{proof}

\begin{remark}[Physical interpretation]
In the language of statistical physics, the tax is a uniform external
field that couples identically to all degrees of freedom.  Such a field
shifts the equilibrium of every mode by the same amount; the relative
structure---which modes are excited, which are suppressed---is
preserved.  This is precisely the content of the neutrality theorem in
\citet{Froeseth2026N}: the tax contracts the opportunity set
homothetically without distorting its shape.
\end{remark}

\subsection{Probability current and detailed balance}

In an untaxed system with no income or consumption, the probability
current is $J_0 = v\pi - D\,\partial\pi/\partial x$.  In the taxed
system, $J_\tau = (v - \tw)\pi - D\,\partial\pi/\partial x$.  The
\emph{difference} is
\begin{equation}\label{eq:current_diff}
  J_\tau - J_0 = -\tw \, \pi(x,t) \,.
\end{equation}
This is a uniform reduction of the probability current, proportional to
the local density.  No new currents are created between states that
did not already have a current; no existing currents are reversed.  The
tax removes probability uniformly, like a spatially uniform decay
rate.

If the untaxed system is in detailed balance ($J_0 = 0$ at steady
state), then the taxed system has $J_\tau = -\tw \pi < 0$: a uniform
leftward current representing the steady drain of wealth.  Detailed
balance is broken, but in the most benign way possible---the system is
driven uniformly, with no state-dependent distortion.  (Note that
detailed balance requires a stationary distribution, which pure GBM does
not possess; the confining mechanisms of \Cref{sec:steady} are needed.
The argument here applies to the stationary regime of such an extended
model.)

\subsection{Robustness of the drift-shift symmetry}\label{sec:robustness}

The preceding results were derived under geometric Brownian motion:
normally distributed log-returns with constant drift and volatility.
We now show that the drift-shift symmetry is robust to two important
generalisations---non-Gaussian return distributions and stochastic
volatility---because it rests on the \emph{tax mechanism}, not on the
distributional form.

\subsubsection{Beyond Gaussian returns}\label{sec:loc_scale}

The neutrality analysis in \citet{Froeseth2026N} shows that
the proportional wealth tax acts as a deterministic multiplicative
scalar on wealth, independently of the return realisation.  After $n$
periods, a taxed investor's wealth is
\begin{equation}\label{eq:mult_sep}
  W_n^{\mathrm{taxed}} = (1 - \tw)^n \; W_n^{\mathrm{untaxed}} \,.
\end{equation}
This identity requires only two properties of the tax: proportionality
(a constant fraction~$\tw$ of all holdings) and universality (the same
rate on every asset).  It makes no reference to the distribution of
returns.  In log-wealth coordinates, the multiplicative factor
becomes an additive constant:
\begin{equation}\label{eq:log_sep}
  x_n^{\mathrm{taxed}} = x_n^{\mathrm{untaxed}} + n \ln(1 - \tw) \,,
\end{equation}
so the tax is a deterministic translation in log-wealth space,
regardless of whether the increments of $x_n^{\mathrm{untaxed}}$ are
Gaussian, Student-$t$, or drawn from any other distribution with
well-defined moments.

\begin{proposition}[Distribution-free drift shift]\label{prop:dist_free}
Let the return process be any continuous-time It\^{o} diffusion
\begin{equation}\label{eq:gen_ito}
  \dd x = \mu(x, t)\, \dd t + \sigma(x, t)\, \dd B_t \,,
\end{equation}
with possibly state-dependent drift $\mu(x,t)$ and diffusion
$\sigma(x,t)$.  Under a proportional wealth tax at rate~$\tw$,
the taxed log-wealth satisfies
\begin{equation}\label{eq:gen_ito_tax}
  \dd x^{\mathrm{taxed}} = \bigl[\mu(x, t) - \tw\bigr]\, \dd t
    + \sigma(x, t)\, \dd B_t \,.
\end{equation}
The corresponding Fokker--Planck equation for the taxed density
$\pi_\tau(x,t)$ is
\begin{equation}\label{eq:fp_gen_tax}
  \frac{\partial \pi_\tau}{\partial t}
  = -\frac{\partial}{\partial x}\bigl[(\mu(x,t) - \tw)\, \pi_\tau\bigr]
    + \frac{1}{2}\frac{\partial^2}{\partial x^2}
      \bigl[\sigma(x,t)^2\, \pi_\tau\bigr] \,.
\end{equation}
The tax modifies only the drift coefficient
$\mu \to \mu - \tw$; the diffusion coefficient
$\sigma(x,t)^2/2$ is unchanged.
\end{proposition}

\begin{proof}
The proportional wealth tax removes a fraction~$\tw$ of wealth per
unit time.  In the wealth-level dynamics
$\dd W = W[\mu_W(x,t)\,\dd t + \sigma_W(x,t)\,\dd B_t]$,
the tax enters as an additional deterministic drain $-\tw W\,\dd t$
in the drift, giving
$\dd W = W[(\mu_W - \tw)\,\dd t + \sigma_W\,\dd B_t]$.
Applying It\^{o}'s lemma to $x = \ln W$ yields
\eqref{eq:gen_ito_tax}, since the It\^{o} correction
$-\sigma_W^2/2$ and the noise term $\sigma_W\,\dd B_t$ are
both independent of~$\tw$.  The Fokker--Planck equation
\eqref{eq:fp_gen_tax} follows from the standard forward
Kolmogorov equation for~\eqref{eq:gen_ito_tax}.
\end{proof}

\begin{remark}[Scope of the distribution-free result]
\Cref{prop:dist_free} shows that the drift-shift structure
$\mu \to \mu - \tw$ is preserved for any It\^{o} diffusion,
including processes with state-dependent drift and volatility
(e.g.\ mean-reverting returns, local volatility models).  The
Gaussian propagator~\eqref{eq:propagator} is specific to the
constant-coefficient case; the drift-shift symmetry is not.  In
particular, the neutrality result of
\Cref{prop:neutrality}---that relative drifts, diffusion ratios,
and optimal portfolio weights are tax-invariant---extends to any
return process of the form~\eqref{eq:gen_ito}, provided the tax
applies proportionally and universally.

The explicit formulas for the Pareto exponent~\eqref{eq:pareto} and
the spectral gap~\eqref{eq:spectral_gap} do use the
constant-coefficient (GBM) assumption, because they depend on the
specific form of the stationary distribution.  The qualitative
results---that the tax steepens the Pareto tail and that relaxation
is slow---hold more broadly, but the exact functional forms would
differ for non-constant coefficients.
\end{remark}

\subsubsection{Stochastic volatility}\label{sec:stochvol}

The GBM assumption also fixes the volatility~$\sigma$ as constant.
Empirically, asset return volatility is time-varying, with
well-documented clustering, mean reversion, and leverage effects
\citep{Heston1993}.  These features are captured by stochastic
volatility models, in which the variance is itself a random process.

We consider the Heston model as a concrete and widely used
example.  The risky asset price and its instantaneous variance
form a pair of coupled diffusions:
\begin{align}
  \frac{\dd S}{S} &= \mu \, \dd t + \sqrt{v_t} \, \dd B_t^{(1)} \,,
  \label{eq:heston_S}\\[3pt]
  \dd v_t &= \lambda(\theta - v_t) \, \dd t
    + \kappa \sqrt{v_t} \, \dd B_t^{(2)} \,,
  \label{eq:heston_v}
\end{align}
where $v_t = \sigma_t^2$ is the instantaneous variance, $\theta > 0$
the long-run mean, $\lambda > 0$ the rate of mean reversion,
$\kappa > 0$ the volatility of variance, and $B_t^{(1)},
B_t^{(2)}$ are standard Brownian motions with
$\mathrm{corr}(\dd B^{(1)}, \dd B^{(2)}) = \rho$.  When $\rho < 0$,
negative returns coincide with rising volatility (the leverage
effect).

The wealth of an investor who allocates a fraction~$w$ to the risky
asset and the remainder to a risk-free asset with rate~$r_f$, subject
to a proportional wealth tax, evolves as
\begin{equation}\label{eq:heston_dW}
  \dd W = W\bigl[r_f + w(\mu - r_f) - \tw\bigr]\,\dd t
    + wW\sqrt{v_t}\,\dd B_t^{(1)} \,.
\end{equation}
The Fokker--Planck equation for the joint density
$\pi(x, v, t)$ of log-wealth $x = \ln W$ and variance~$v$ is
two-dimensional.  Crucially, the tax rate~$\tw$ appears only in the
drift of~$x$---the dynamics of $v$ (\Cref{eq:heston_v}) are
entirely independent of the tax.  This is the two-dimensional
analogue of the drift-shift symmetry: the tax acts along the wealth
axis only, leaving the volatility state untouched.

\begin{proposition}[Portfolio neutrality under stochastic volatility]
\label{prop:sv_neutral}
Under the Heston model
\eqref{eq:heston_S}--\eqref{eq:heston_v} with CRRA
preferences $U(C) = C^{1-\gamma}/(1-\gamma)$ and a proportional
wealth tax on all assets, the optimal portfolio weight $w^*$ is
independent of the wealth tax rate~$\tw$.
\end{proposition}

\begin{proof}
Under CRRA utility, the value function admits the separable form
$J(W, v, t) = W^{1-\gamma} f(v, t)/(1-\gamma)$, where $f > 0$
encodes the dependence on the volatility state and the investment
horizon.  The first-order condition for the portfolio weight in
the Hamilton--Jacobi--Bellman equation yields (see
\Cref{app:stochvol} for the full derivation)
\begin{equation}\label{eq:wstar_sv}
  w^* = \underbrace{\frac{\mu - r_f}{\gamma v}}_{\text{myopic}}
    + \underbrace{\frac{f_v}{f}
      \cdot \frac{\kappa\rho}{\gamma}}_{\text{hedging}} \,.
\end{equation}
The myopic demand depends on the excess return, risk aversion,
and current variance---not on wealth or the tax rate.  The hedging
demand---the intertemporal component identified by
\citet{Merton1973}---depends on the ratio $f_v/f$, where $f(v,t)$ satisfies a PDE
obtained by substituting the separable value function into the HJB
equation.  In this PDE, the tax rate~$\tw$ appears only in a
term $(1-\gamma)(r_f - \tw)f$ that shifts the effective discount
rate but does not interact with~$v$.  Using the standard
exponential-affine ansatz $f(v,t) = \exp(A(t) + B(t)v)$
\citep{ChackoViceira2005}, the Riccati equation for $B(t)$ collects
only the $v$-dependent terms and is therefore independent
of~$\tw$.  Since $f_v/f = B(t)$, the hedging demand is
tax-invariant.  Consequently, both components of $w^*$ are
independent of~$\tw$.
\end{proof}

\begin{corollary}[General Markov diffusions]\label{cor:markov}
The result extends to any Markov diffusion model with $K$ risky
assets and $M$ state variables
$\mathbf{X}_t = (X_1, \ldots, X_M)^\top$,
in which the expected returns $\mu_i(\mathbf{X})$, volatilities
$\sigma_{ij}(\mathbf{X})$, and state dynamics
$a_m(\mathbf{X})$, $b_{mj}(\mathbf{X})$ depend on the state but
not on the investor's wealth.  Under CRRA preferences, the value
function separates as
$J = W^{1-\gamma} f(\mathbf{X}, t)/(1-\gamma)$,
and the optimal portfolio weights
\[
  \mathbf{w}^* = \frac{1}{\gamma}\mathbf{V}(\mathbf{X})^{-1}
    \bigl(\boldsymbol{\mu}(\mathbf{X}) - r_f\mathbf{1}\bigr)
    + \frac{1}{\gamma}\mathbf{V}(\mathbf{X})^{-1}
      \boldsymbol{\Phi}(\mathbf{X})\,
      \frac{\nabla_{\mathbf{X}} f}{f}
\]
are independent of the wealth tax rate~$\tw$.  This encompasses
the Heston, Hull--White, SABR, and affine term structure models
as special cases (see \citet{Froeseth2026E} for the full proof).
\end{corollary}

\begin{remark}[Two mechanisms, one conclusion]
The location-scale result (\Cref{sec:loc_scale}) and the
stochastic volatility result rest on different mechanisms.  The
former uses the algebraic structure of the tax (multiplicative
separability) and requires no utility specification.  The latter
requires CRRA preferences but places no restriction on the return
distribution---Heston returns are not in the location-scale
family.  Together, they show that the drift-shift symmetry and its
portfolio neutrality consequence are robust across a wide class of
models.  \Cref{app:geometry} develops a geometric interpretation
that unifies both mechanisms in the language of fiber bundles and
Galilean symmetry.
\end{remark}

\section{Distortion channels as symmetry breaking}\label{sec:channels}

The neutrality result of \Cref{sec:neutrality} rests on the wealth tax
entering as a uniform, state-independent drift shift.  Each of the
distortion channels identified in \citet{Froeseth2026E} breaks this
condition in a specific way.  We now classify them by the type of
modification they introduce into the Fokker--Planck equation.

\begin{remark}[Physical intuition: gravity versus friction]
The physical analogies are more than formal.  The neutral wealth tax
acts as a \emph{uniform gravitational field} on the wealth distribution:
it pulls every particle (investor) downward at the same rate,
regardless of composition, without altering the thermal fluctuations
(volatility).  Gravity modifies the drift but not the diffusion---precisely
the drift-shift symmetry.  Because the tax payment $\tw W$ is
deterministic given wealth, it enters the Langevin equation as a
force, not as noise; a deterministic force can only modify the drift
coefficient.

Each distortion channel introduces something beyond gravity.
Liquidity frictions are literally \emph{friction}: forced selling into
illiquid markets dissipates wealth in a manner that depends on the
state and couples to both drift and diffusion.  Book-value assessment
introduces an \emph{anisotropic} field that pulls different assets at
different rates.  Migration creates an \emph{absorbing boundary}---a
cliff edge in the potential landscape.  These analogies, summarised
alongside the formal modifications below, may help readers carry a
physical picture through the classification.
\end{remark}

\subsection{Channel 1: Book-value assessment}\label{sec:ch1}

When assets are taxed at book value rather than market value, different
assets attract different effective tax rates.  If asset~$i$ has a
book-to-market ratio~$\beta_i$, the effective tax rate is
$\tw^{(i)} = \tw \cdot \beta_i$.  The drift shift becomes
asset-dependent:
\begin{equation}\label{eq:ch1}
  v_i \;\to\; v_i - \tw \beta_i \,.
\end{equation}
The drift-shift transformation is no longer uniform: assets with low
book-to-market ratios (growth stocks, intangible-heavy firms) are
taxed less than assets with high book-to-market ratios (value stocks,
asset-heavy firms).  In the Fokker--Planck equation, the drift
coefficient becomes state-dependent through the portfolio composition,
since the portfolio-level effective tax rate depends on which assets the
investor holds.

\textbf{The leverage amplification.}
The anisotropic drift is amplified by the asymmetric treatment of debt
in the net wealth tax base.  Assets are assessed at book value
($\beta_i < 1$ for underassessed assets such as real estate), but
liabilities are deducted at face value ($\beta_{\mathrm{debt}} = 1$).
An investor with assets of market value~$A$ assessed at
$\beta A$ and debt~$D$ has:
\begin{equation}\label{eq:leverage_wedge}
  W_{\mathrm{tax}} = \beta A - D \,,
  \qquad W_{\mathrm{market}} = A - D \,,
\end{equation}
so the ratio $W_{\mathrm{tax}} / W_{\mathrm{market}} = (\beta A - D)/(A - D)$
decreases with leverage and can become negative even when
$W_{\mathrm{market}} > 0$.  This creates an incentive to lever up:
borrowing against underassessed assets reduces the tax base
disproportionately.  In the Fokker--Planck equation, the effective
drift depends not only on portfolio composition (through~$\beta_i$)
but also on the leverage ratio, coupling the debt decision to the
wealth dynamics.  Moreover, leveraged positions amplify the effective
volatility of net wealth---connecting the assessment channel to the
diffusion modification of Channel~2 below.

\textbf{Symmetry broken:} Uniformity of $\mathcal{T}_\tau$ across
assets.  The tax now couples differently to different degrees of
freedom, distorting relative returns and portfolio choice.  The
asymmetric treatment of debt amplifies the distortion for leveraged
investors.

\subsection{Channel 2: Liquidity frictions}\label{sec:ch2}

If the investor must sell assets to pay the tax and faces transaction
costs or illiquidity, the effective tax rate depends on the liquidity of
the portfolio.  In the extreme, an investor holding only illiquid
assets faces a higher effective burden than one holding liquid assets.
The Fokker--Planck equation acquires a \emph{state-dependent drift
modification}:
\begin{equation}\label{eq:ch2}
  v \;\to\; v - \tw - c(W, \ell) \,,
\end{equation}
where $c(W, \ell)$ is a friction cost that depends on wealth~$W$ (or
log-wealth~$x$) and a liquidity parameter~$\ell$.  For investors with
wealth concentrated in illiquid assets, $c$ is large; for liquid
portfolios, $c \approx 0$.

Moreover, liquidity frictions can modify the \emph{diffusion}
coefficient if forced selling at unfavourable prices increases effective
volatility:
\begin{equation}\label{eq:ch2_diff}
  D \;\to\; D + \Delta D(W, \ell) \,,
\end{equation}
introducing a state-dependent diffusion that couples the tax to the
stochastic structure of returns.

\textbf{Symmetry broken:} State-independence and drift-only coupling.
The tax now modifies both drift and diffusion, and does so differently
for different investors.

\subsection{Channel 3: Forced dividend extraction}\label{sec:ch3}

When firm owners extract dividends to pay the wealth tax, the firm's
capital stock is reduced.  If the firm faces financing frictions
(limited credit, costly equity issuance), this extraction reduces
investment and growth.  The wealth tax then \emph{couples the
investor's dynamics to the firm's dynamics}:
\begin{equation}\label{eq:ch3}
  \mu \;\to\; \mu(K) \quad \text{where} \quad
  \frac{\dd K}{\dd t} = f(K) - \delta K - \tw W \,,
\end{equation}
and $K$ is the firm's capital, $f(K)$ is the production function,
$\delta$ is depreciation.  The expected return~$\mu$ is no longer a
constant but depends on the firm's state, which in turn depends on the
cumulative tax extracted.  The Fokker--Planck equation acquires a
\emph{memory} through the coupling to the slow variable~$K$.

\textbf{Symmetry broken:} Constancy of drift.  The drift is endogenous,
coupled to firm-level dynamics that evolve on a different timescale.

\subsection{Channel 4: Migration}\label{sec:ch4}

If investors can emigrate to escape the wealth tax, the system acquires an
\emph{absorbing boundary} in wealth space.  An investor with wealth
$W > W^*$ (where $W^*$ is the threshold at which the tax burden exceeds
the cost of emigration) may exit the system entirely.  The
Fokker--Planck equation is modified by a sink term:
\begin{equation}\label{eq:ch4}
  \frac{\partial \pi}{\partial t}
  = -v_\tau \frac{\partial \pi}{\partial x}
    + D \frac{\partial^2 \pi}{\partial x^2}
    - \gamma(x)\, \pi \,,
\end{equation}
where $\gamma(x)$ is a loss rate that increases sharply for
$x > x^* = \ln W^*$.  This drains probability from the upper tail of
the distribution, depleting the wealthiest investors.

\textbf{Symmetry broken:} The system is no longer closed.  The tax
induces a state-dependent outflow that preferentially removes
high-wealth investors, truncating the distribution.

\subsection{Channel 5: Market impact}\label{sec:ch5}

If forced asset sales to pay the tax move prices, the tax introduces a
\emph{flow-dependent} nonlinearity.  The square-root impact
law---a robust empirical regularity in market
microstructure---implies that aggregate selling pressure~$F$
depresses prices by an amount proportional to~$\sqrt{F}$.  When
combined with the low aggregate elasticity of demand documented by
\citet{GabaixKoijen2022}, even moderate tax-induced flows can
generate sizable price dislocations.  The effective return becomes
\begin{equation}\label{eq:ch5}
  \mu \;\to\; \mu - \alpha\sqrt{F(\tw, \mathcal{N})} \,,
\end{equation}
where $F$ depends on the tax rate and the number of investors selling
simultaneously.  This introduces a \emph{collective effect}: the drift
reduction experienced by each investor depends on the aggregate
behaviour of all investors.

In the Fokker--Planck equation, this appears as a nonlinear,
self-consistent drift:
\begin{equation}\label{eq:ch5_fp}
  v \;\to\; v - \tw
    - \alpha\sqrt{\tw \int W\, p(W,t)\, \dd W} \,.
\end{equation}
The drift depends on the first moment of the distribution itself,
creating a feedback loop between the individual dynamics and the
aggregate state.

\textbf{Symmetry broken:} Linearity and investor-independence.  The tax
couples each investor's dynamics to the aggregate distribution, introducing
mean-field interactions absent from the untaxed system.

\subsection{Summary of symmetry-breaking mechanisms}

\begin{center}
\renewcommand{\arraystretch}{1.4}
\begin{tabular}{@{}llll@{}}
\toprule
\textbf{Channel} & \textbf{FP modification} & \textbf{Symmetry broken} & \textbf{Physics analogue} \\
\midrule
Neutral tax & $v \to v - \tw$ & None & Uniform external field \\
1. Book value & $v_i \to v_i - \tw\beta_i$ & Uniformity across assets
  & Anisotropic field \\
2. Liquidity & $v, D \to v(W), D(W)$ & State-independence
  & State-dependent friction \\
3. Dividends & $\mu \to \mu(K(t))$ & Drift constancy
  & Coupled slow variable \\
4. Migration & $+$ sink $\gamma(x)\pi$ & Closed system
  & Absorbing boundary \\
5. Market impact & $v \to v(\langle W \rangle)$ & Investor-independence
  & Mean-field interaction \\
\bottomrule
\end{tabular}
\end{center}

\Cref{fig:channels} illustrates how each channel deforms the
propagator relative to the neutral (pure translation) case.

\begin{figure}[H]
\centering
\begin{tikzpicture}[
  every axis/.style={
    width=0.46\textwidth,
    height=0.32\textwidth,
    axis lines=left,
    xlabel={},
    ylabel={},
    ymin=0,
    ymax=0.75,
    xmin=-1,
    xmax=8,
    samples=150,
    domain=-1:8,
    ticks=none,
    every axis plot/.append style={thick},
    clip=true,
  },
  panel label/.style={
    font=\small\bfseries,
    anchor=north west,
  },
  centre line/.style={densely dotted, thin},
]

\begin{axis}[at={(0,10.0cm)}, name=pA]
  \draw[blue!60, centre line] (axis cs:4.6,0) -- (axis cs:4.6,0.43);
  \draw[red!60, centre line] (axis cs:3.6,0) -- (axis cs:3.6,0.43);
  \addplot[blue!70!black, solid]
    {1.3/(sqrt(2*3.14159*1.6))*exp(-(x-4.6)^2/(2*1.6))};
  \addplot[red!70!black, solid]
    {1.3/(sqrt(2*3.14159*1.6))*exp(-(x-3.6)^2/(2*1.6))};
  \draw[-{Stealth[length=2mm]}, black!60, thick]
    (axis cs:4.6,0.47) -- (axis cs:3.6,0.47);
  \node[font=\scriptsize, black!60, above] at (axis cs:4.1,0.47)
    {$\tau_w t$};
  \node[panel label] at (axis cs:-0.8,0.72)
    {(a) Neutral tax};
  \node[font=\scriptsize, black!50, align=center] at (axis cs:1.2,0.55)
    {same shape,\\shifted centre};
\end{axis}

\begin{axis}[at={(7.5cm,10.0cm)}, name=pB]
  \draw[blue!60, centre line] (axis cs:4.6,0) -- (axis cs:4.6,0.43);
  \draw[red!60, centre line] (axis cs:3.0,0) -- (axis cs:3.0,0.28);
  \draw[red!60, centre line] (axis cs:4.3,0) -- (axis cs:4.3,0.28);
  \draw[-{Stealth[length=2mm]}, black!60]
    (axis cs:4.6,0.50) -- (axis cs:4.3,0.50);
  \node[font=\scriptsize, black!60, above] at (axis cs:4.45,0.50)
    {$\tau_1$};
  \draw[-{Stealth[length=2mm]}, black!60]
    (axis cs:4.6,0.57) -- (axis cs:3.0,0.57);
  \node[font=\scriptsize, black!60, above] at (axis cs:3.8,0.57)
    {$\tau_2$};
  \addplot[blue!70!black, solid]
    {1.3/(sqrt(2*3.14159*1.6))*exp(-(x-4.6)^2/(2*1.6))};
  \addplot[red!70!black, solid]
    {1.3*(0.5/(sqrt(2*3.14159*1.2))*exp(-(x-3.0)^2/(2*1.2))
    +0.5/(sqrt(2*3.14159*1.2))*exp(-(x-4.3)^2/(2*1.2)))};
  \node[panel label] at (axis cs:-0.8,0.72)
    {(b) Book-value basis};
  \node[font=\scriptsize, black!50, align=center] at (axis cs:1.2,0.55)
    {anisotropic shift\\splits modes};
\end{axis}

\begin{axis}[at={(0,5.0cm)}, name=pC]
  \draw[blue!60, centre line] (axis cs:4.6,0) -- (axis cs:4.6,0.43);
  \draw[red!60, centre line] (axis cs:3.6,0) -- (axis cs:3.6,0.33);
  \draw[red!50, centre line] (axis cs:3.6-1.67,0) -- (axis cs:3.6-1.67,0.16);
  \draw[red!50, centre line] (axis cs:3.6+1.67,0) -- (axis cs:3.6+1.67,0.16);
  \draw[-{Stealth[length=2mm]}, black!60, thick]
    (axis cs:4.6,0.47) -- (axis cs:3.6,0.47);
  \node[font=\scriptsize, black!60, above] at (axis cs:4.1,0.47)
    {$\tau_w t$};
  \draw[{Stealth[length=1.5mm]}-{Stealth[length=1.5mm]}, red!50, thin]
    (axis cs:3.6-1.67,0.04) -- (axis cs:3.6+1.67,0.04);
  \node[font=\scriptsize, red!50, fill=white, inner sep=0.5pt]
    at (axis cs:3.6,0.04) {$\sigma_{\mathrm{eff}}$};
  \addplot[blue!70!black, solid]
    {1.3/(sqrt(2*3.14159*1.6))*exp(-(x-4.6)^2/(2*1.6))};
  \addplot[red!70!black, solid]
    {1.3/(sqrt(2*3.14159*2.8))*exp(-(x-3.6)^2/(2*2.8))};
  \node[panel label] at (axis cs:-0.8,0.72)
    {(c) Liquidity frictions};
  \node[font=\scriptsize, black!50, align=center] at (axis cs:1.2,0.55)
    {broadened:\\$D_{\mathrm{eff}} > D$};
\end{axis}

\begin{axis}[at={(7.5cm,5.0cm)}, name=pD]
  \draw[blue!60, centre line] (axis cs:4.6,0) -- (axis cs:4.6,0.43);
  \draw[red!60, centre line] (axis cs:3.4,0) -- (axis cs:3.4,0.43);
  \draw[-{Stealth[length=2mm]}, black!60, thick]
    (axis cs:4.6,0.47) -- (axis cs:3.4,0.47);
  \node[font=\scriptsize, black!60, above] at (axis cs:4.0,0.47)
    {$\tau_w t + \delta v \cdot t$};
  \addplot[blue!70!black, solid]
    {1.3/(sqrt(2*3.14159*1.6))*exp(-(x-4.6)^2/(2*1.6))};
  \addplot[red!70!black, solid]
    {1.3/(sqrt(2*3.14159*1.6))*exp(-(x-3.4)^2/(2*1.6))
     * max(0, 1 + 0.4*(x-3.4)/1.265)};
  \node[panel label] at (axis cs:-0.8,0.72)
    {(d) Dividend extraction};
  \node[font=\scriptsize, black!50, align=center] at (axis cs:1.2,0.55)
    {asymmetric:\\$v = v(t)$};
\end{axis}

\begin{axis}[at={(0,0)}, name=pE]
  \draw[blue!60, centre line] (axis cs:4.6,0) -- (axis cs:4.6,0.43);
  \draw[red!60, centre line] (axis cs:3.6,0) -- (axis cs:3.6,0.43);
  \draw[-{Stealth[length=2mm]}, black!60, thick]
    (axis cs:4.6,0.47) -- (axis cs:3.6,0.47);
  \node[font=\scriptsize, black!60, above] at (axis cs:4.1,0.47)
    {$\tau_w t$};
  \draw[black!60, dashed, thin] (axis cs:5.5,0) -- (axis cs:5.5,0.65);
  \node[font=\scriptsize, black!60, above] at (axis cs:5.5,0.65) {$x^*$};
  \addplot[blue!70!black, solid]
    {1.3/(sqrt(2*3.14159*1.6))*exp(-(x-4.6)^2/(2*1.6))};
  \addplot[red!70!black, solid]
    {1.3/(sqrt(2*3.14159*1.6))*exp(-(x-3.6)^2/(2*1.6))
     * (x < 5.5 ? 1 : exp(-5*(x-5.5)^2))};
  \node[panel label] at (axis cs:-0.8,0.72)
    {(e) Migration};
  \node[font=\scriptsize, black!50, align=center] at (axis cs:1.2,0.55)
    {absorbing\\boundary at $x^*$};
\end{axis}

\begin{axis}[at={(7.5cm,0)}, name=pF]
  \draw[blue!60, centre line] (axis cs:4.6,0) -- (axis cs:4.6,0.43);
  \draw[red!60, centre line] (axis cs:2.6,0) -- (axis cs:2.6,0.47);
  \draw[-{Stealth[length=2mm]}, black!60, thick]
    (axis cs:4.6,0.50) -- (axis cs:2.6,0.50);
  \node[font=\scriptsize, black!60, above] at (axis cs:3.6,0.50)
    {$\tau_w t + \Delta p$};
  \addplot[blue!70!black, solid]
    {1.3/(sqrt(2*3.14159*1.6))*exp(-(x-4.6)^2/(2*1.6))};
  \addplot[red!70!black, solid]
    {1.3/(sqrt(2*3.14159*1.3))*exp(-(x-2.6)^2/(2*1.3))};
  \node[panel label] at (axis cs:-0.8,0.72)
    {(f) Market impact};
  \node[font=\scriptsize, black!50, align=center] at (axis cs:1.0,0.55)
    {extra shift from\\price feedback};
\end{axis}

\end{tikzpicture}
\caption{How each distortion channel deforms the wealth distribution
  relative to the neutral case.  Blue: untaxed propagator.  Red: taxed
  propagator.  (a)~Neutral tax: pure translation, no deformation.
  (b)~Book-value assessment: different assets shift by different
  amounts, splitting the distribution into modes.  (c)~Liquidity
  frictions: the distribution broadens (increased effective $D$) as
  forced selling at unfavourable prices adds noise.  (d)~Dividend
  extraction: the drift becomes time-dependent as firm capital erodes,
  introducing asymmetry.  (e)~Migration: high-wealth investors exit the
  system above a threshold~$x^*$, truncating the right tail.
  (f)~Market impact: the collective selling pressure shifts the
  distribution further left than the tax alone would predict, with
  additional compression from price feedback.}
\label{fig:channels}
\end{figure}
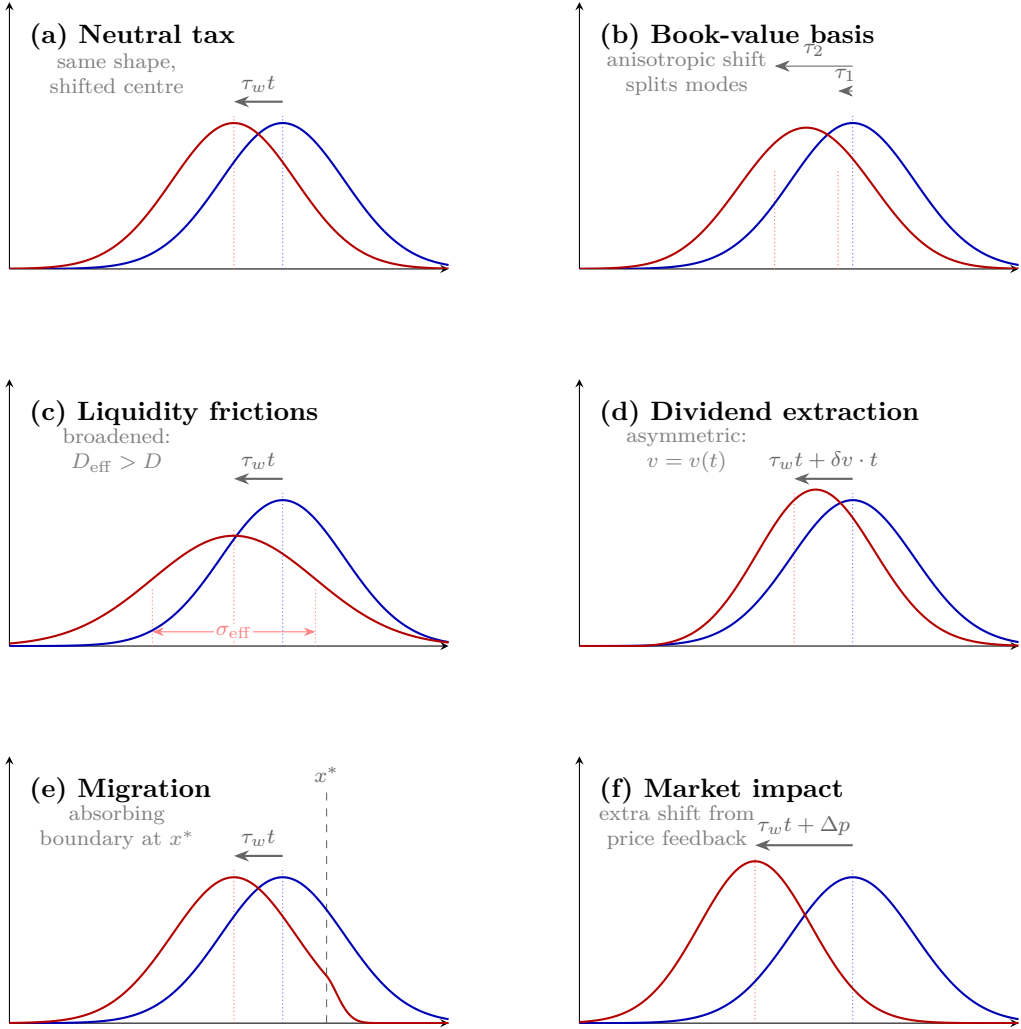

\section{Toward steady-state distributions}\label{sec:steady}

\subsection{The problem with pure GBM}

The Fokker--Planck equation \eqref{eq:fp_x} with constant $v$ and $D$
has no stationary solution on $(-\infty, \infty)$: the distribution
drifts and spreads indefinitely.  This reflects the well-known fact that
geometric Brownian motion produces a log-normal distribution whose mean
and variance both grow without bound.

For a wealth distribution to reach a steady state, the system must have
sources and sinks that balance drift and diffusion.

\subsection{Income and consumption as source and sink}

In a realistic model, investors receive income and consume.  The simplest
extension adds an income flux $\lambda$ (in units of log-wealth per
time) and a consumption rate $c$, yielding
\begin{equation}\label{eq:income}
  \dd x = \left(v + \frac{\lambda}{W} - c\right) \dd t
    + \sigma \, \dd B_t \,,
\end{equation}
where the $\lambda/W$ term reflects that a fixed income~$\lambda$ has a
larger effect on log-wealth when $W$ is small.  This creates a
state-dependent drift that pushes low-wealth investors upward (income
dominates) and allows high-wealth investors to drift further right (returns
dominate).  The resulting Fokker--Planck equation has a confining effect
at low wealth and can admit a stationary solution.

\subsection{Structure of the stationary distribution}

Setting $\partial\pi/\partial t = 0$ in the modified Fokker--Planck
equation with income and consumption yields a second-order ODE for
$\pi_{\mathrm{ss}}(x)$.  While the general solution depends on the
specific functional forms, the qualitative structure is well established
in the econophysics literature \citep{Yakovenko2009}:
\begin{itemize}
  \item \textbf{Bulk} (low to moderate wealth): Approximately
    exponential (Boltzmann--Gibbs) distribution, $\pi \propto
    e^{-x/T}$, where $T$ is an effective temperature related to the
    average wealth.
  \item \textbf{Tail} (high wealth): Power-law (Pareto) distribution,
    $p(W) \propto W^{-\alpha}$, arising from the multiplicative nature
    of returns at high wealth where income is negligible relative to
    capital gains.
\end{itemize}

The wealth tax modifies the Pareto exponent.  When GBM is supplemented
by an additive income or redistribution term---turning the dynamics into
a Kesten process---the stationary distribution develops a Pareto tail
\citep{Kesten1973,BouchaudMezard2000}; see \citet{Gabaix2009} for a
survey of the random-growth mechanism and its applications.  The tail exponent satisfies
(see \Cref{app:pareto} for the derivation)
\begin{equation}\label{eq:pareto}
  \alpha(\tw) = 1 - \frac{2(\mu - \tw)}{\sigma^2} \,.
\end{equation}
The condition $\alpha > 0$ requires $\mu - \tw < \sigma^2/2$, i.e.\
the log-wealth drift must be negative for a stationary distribution to
exist.  Increasing the tax rate increases~$\alpha$, making the tail
steeper (less inequality): the tax reduces the effective drift that
sustains the power law, compressing the upper tail of the wealth
distribution.

\begin{remark}[Neutrality and the distribution]
There is no contradiction between the neutrality result (individual
portfolio choice unchanged) and the distributional effect (Pareto
exponent changes).  Neutrality is a statement about each investor's
optimisation problem; the distributional effect is a statement about the
ensemble.  The tax is neutral in the sense that no investor wants to
change their portfolio, but the aggregate distribution shifts because
all investors are uniformly poorer.
\end{remark}

\begin{remark}[Distinction from Bouchaud--M\'{e}zard redistribution]
The proportional wealth tax studied here is structurally different from
the redistribution mechanism in \citet{BouchaudMezard2000}.  Their
model features a mean-field exchange term
$J \sum_j (w_j - w_i)$ that transfers wealth between agents---an
off-diagonal coupling in the multi-agent Langevin system.  The
proportional wealth tax, by contrast, modifies each investor's own
drift uniformly: $v \to v - \tw$, a diagonal perturbation that
leaves the diffusion and all inter-investor couplings unchanged.
These are distinct modifications of the Fokker--Planck equation.
The exchange mechanism confines the distribution through
inter-agent interactions; the wealth tax confines it through a
shift in each investor's individual growth rate.  The drift-shift
symmetry that underlies our neutrality result has no analogue in
the exchange framework.
When agents additionally have heterogeneous growth rates,
\citet{BernardBouchaudLeDoussal2026} show that the Bouchaud--Mézard
model exhibits a localisation--delocalisation phase transition at a
critical redistribution rate, with an intermediate partially localised
phase governed by a Random Energy Model analogy.
\end{remark}

\subsection{Ergodicity and the Pareto condition}\label{sec:ergodicity}

The Pareto exponent \eqref{eq:pareto} admits a revealing reinterpretation
through the lens of ergodicity \citep{Peters2019}.  Under geometric
Brownian motion, the ensemble-average growth rate of wealth
$\E[\dd W/W]/\dd t = \mu$ differs from the time-average (almost-sure)
growth rate of log-wealth,
\begin{equation}\label{eq:time_avg}
  g_{\mathrm{time}} = \mu - \frac{\sigma^2}{2} \,,
\end{equation}
whenever $\sigma > 0$.  The gap between them is $\sigma^2/2 = D$,
exactly the diffusion coefficient in the Fokker--Planck equation.  This
is the signature of non-ergodicity for multiplicative processes: the
typical trajectory and the ensemble average diverge, and the divergence
is governed by the same parameter that controls the spreading of the
wealth distribution
\citep[see][for ergodicity in Brownian systems]{Zwanzig2001}\citep[and][for
quasi-non-ergodicity in wealth dynamics]{BouchaudFarmer2021}.

The wealth tax reduces both averages uniformly: the ensemble-average
growth rate becomes $\mu - \tw$ and the time-average becomes
\begin{equation}\label{eq:time_avg_taxed}
  v_\tau \;=\; \mu - \tw - \frac{\sigma^2}{2} \,,
\end{equation}
which is the drift of taxed log-wealth that appears throughout the
preceding sections.  The Pareto exponent can now be written directly
in terms of this time-average growth rate:
\begin{equation}\label{eq:pareto_ergodic}
  \alpha(\tw)
    = -\,\frac{2\,v_\tau}{\sigma^2}
    = -\,\frac{v_\tau}{D} \,.
\end{equation}
This is a drift-to-diffusion ratio---the analogue of a P\'eclet number
in transport physics---measuring the strength of directed motion
relative to diffusive spreading.  A physicist will recognise it
immediately: the Pareto tail steepness is set by the competition between
drift (which concentrates wealth) and diffusion (which disperses it).

Three consequences follow.

First, \emph{the condition for a stationary distribution
($\alpha > 0$) is equivalent to} $v_\tau < 0$: the time-average
growth rate of taxed log-wealth must be negative.  This is not a
pathology.  It means the typical investor's wealth is shrinking over
time, even as the ensemble mean grows (because the ensemble average
is dominated by a few lucky trajectories in the upper tail).  The
tension between multiplicative growth for the fortunate few and
mean-reversion via the additive Kesten term is precisely what
creates and sustains the Pareto tail.

Second, \emph{the wealth tax steepens the Pareto tail by making
$v_\tau$ more negative.}  This is the distributional counterpart of
the neutrality result: the tax does not change any investor's
portfolio (neutrality), but by shifting $v_\tau$ downward it tips
the drift--diffusion balance further toward diffusion dominance,
compressing the upper tail of the wealth distribution.

Third, \emph{the non-ergodicity gap $\sigma^2/2 = D$ is invariant
under the wealth tax.}  Since the tax modifies only the drift, not
the diffusion coefficient, it cannot alter the fundamental
non-ergodic character of the dynamics.  The gap between what the
ensemble predicts and what the typical investor experiences remains
unchanged.  This invariance will reappear in the analysis of
relaxation times, where $D$ sets the timescale over which the
distribution approaches its new steady state.

Finally, the neutrality result acquires a clean ergodic
restatement.  Consider a portfolio of $N$ assets with
time-average growth rates $g_i = \mu_i - \sigma_i^2/2$.  The wealth
tax transforms each rate to $g_i - \tw$, preserving the ranking:
if $g_i > g_j$ before tax, then $g_i - \tw > g_j - \tw$ after.
Since long-run portfolio choice depends on the ranking of
time-average growth rates (the asset that compounds fastest almost
surely dominates), the tax is neutral for portfolio selection.

\subsection{Relaxation dynamics and the spectral gap}\label{sec:relaxation}

The ergodic reinterpretation tells us \emph{where} the wealth
distribution converges to---a Pareto tail with exponent~$\alpha$.  It
does not tell us \emph{how fast}.  After a tax change shifts the drift,
the distribution relaxes toward a new steady state.  The speed of this
relaxation is governed by the \emph{spectral gap} of the
Fokker--Planck operator: the magnitude of the second eigenvalue of the
Kolmogorov forward operator $\mathcal{A}^*$
\citep[for the general eigenvalue theory, see][Ch.~6]{LiviPoliti2017}.

\citet{GabaixEtAl2016} prove that for the random growth process
$\dd x = \mu\, \dd t + \sigma\, \dd B_t$ with reflecting barrier at
$x = 0$ and demographic turnover at rate~$\delta$ (investors die and are
replaced at a reference level), the cross-sectional distribution
$p(x,t)$ converges exponentially to its stationary
distribution~$p_\infty(x)$:
\[
  \|p(\cdot, t) - p_\infty(\cdot)\| \;\sim\; k\, e^{-\lambda t} \,,
\]
where $\|\cdot\|$ is the $L^1$ (total variation) norm and the rate
of convergence is
\begin{equation}\label{eq:spectral_gap}
  \lambda \;=\; \frac{v_\tau^2}{2\sigma^2}
    \,\mathbf{1}_{\{v_\tau < 0\}} \;+\; \delta \,.
\end{equation}
The indicator function $\mathbf{1}_{\{v_\tau < 0\}}$ is the key: the
drift-dependent term contributes to convergence \emph{only} when the
time-average growth rate is negative ($v_\tau < 0$), so that the drift
pushes probability back toward the barrier.  When $v_\tau > 0$,
the drift carries wealth away from the barrier, and only demographic
turnover~$\delta$ drives convergence.  A physicist will recognise
two distinct relaxation mechanisms: confining drift (analogous to a
restoring force) and population renewal (analogous to coupling to a
thermal reservoir).

The corresponding half-life is
\begin{equation}\label{eq:half_life}
  t_{1/2} \;=\; \frac{\ln 2}{\lambda} \,.
\end{equation}

Three features of \eqref{eq:spectral_gap} deserve emphasis.

First, \emph{the Pareto exponent and the relaxation rate are linked.}
When $v_\tau < 0$, the Pareto exponent from
\Cref{sec:ergodicity} is $\alpha = -2v_\tau/\sigma^2$, so the
drift-dependent contribution to~$\lambda$ can be written
\begin{equation}\label{eq:lambda_alpha}
  \lambda \;=\; \frac{\alpha^2 \sigma^2}{8} + \delta \,.
\end{equation}
The two key observables of the wealth distribution---the tail steepness
($\alpha$) and the convergence rate ($\lambda$)---are related through a
single formula.  A steeper tail (larger~$\alpha$) implies faster
convergence, because a stronger confining drift (more
negative~$v_\tau$) both steepens the Pareto tail and accelerates the
return to equilibrium.

Second, \emph{there is a critical threshold at $v_\tau = 0$}
(equivalently, $\tw = \mu - \sigma^2/2$, or $\alpha = 0$).  Below
this threshold, the drift-dependent mechanism is inactive: convergence
relies entirely on demographic turnover at rate~$\delta$.  Above it,
both mechanisms operate.  This is the relaxation-dynamics counterpart of
the steady-state result that the Pareto tail exists only for
$\alpha > 0$.

Third, \emph{for realistic parameters, relaxation is slow.}  With
$\sigma = 0.30$ (cross-sectional wealth volatility), $\mu = 0.08$
(expected return), $\delta = 1/30$ (generational turnover), and a
tax rate $\tw = 2\%$: the taxed drift is $v_\tau = 0.08 - 0.02 - 0.045
= 0.015 > 0$.  Since $v_\tau$ remains positive, the drift-dependent
term is inactive, giving $\lambda = \delta = 0.033$ and
$t_{1/2} \approx 21$~years.  Even at $\tw = 5\%$, $v_\tau = -0.015$
and the drift contribution is only $0.015^2/(2 \times 0.09) \approx
0.001$, barely altering the half-life.  Only at much higher tax rates
(or lower expected returns) does the drift mechanism materially
accelerate convergence.

This is the wealth-tax application of \citeauthor{GabaixEtAl2016}'s
central finding: the baseline random growth model generates inherently
slow transition dynamics.  The wealth tax changes \emph{where} the
distribution converges to---the steady-state Pareto exponent~$\alpha$
responds immediately to the new drift---but it has much less effect on
\emph{how fast} it gets there.  Short-run distributional effects of a
tax change may therefore look very different from the long-run
steady state, with the transition stretching over decades.
\Cref{fig:relaxation} illustrates this for three tax rates.

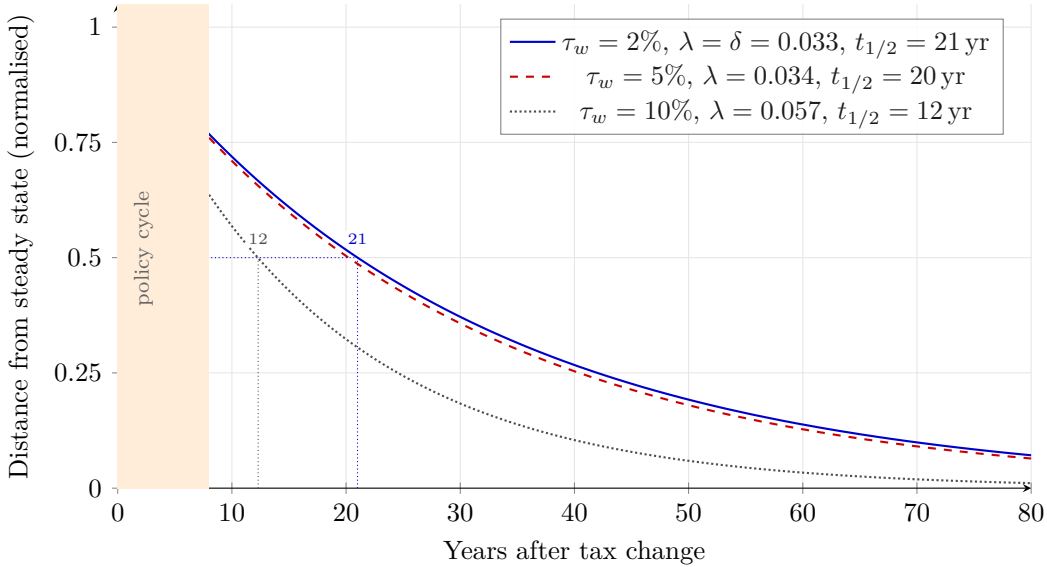
\begin{figure}[H]
\centering
\begin{tikzpicture}
\begin{axis}[
  width=0.85\textwidth,
  height=0.50\textwidth,
  xlabel={Years after tax change},
  ylabel={Distance from steady state (normalised)},
  domain=0:80,
  samples=200,
  ymin=0,
  ymax=1.05,
  xmin=0,
  xmax=80,
  legend style={
    at={(0.97,0.97)},
    anchor=north east,
    font=\small,
    draw=black!50,
    fill=white,
    fill opacity=0.9,
  },
  every axis plot/.append style={thick},
  axis lines=left,
  clip=true,
  tick label style={font=\small},
  label style={font=\small},
  ytick={0, 0.25, 0.5, 0.75, 1.0},
  xtick={0, 10, 20, 30, 40, 50, 60, 70, 80},
  grid=major,
  grid style={gray!20},
]

\addplot[blue!80!black, solid]
  {exp(-0.033*x)};
\addlegendentry{$\tau_w = 2\%$, $\lambda = \delta = 0.033$, $t_{1/2} = 21$\,yr}

\addplot[red!80!black, dashed]
  {exp(-0.0343*x)};
\addlegendentry{$\tau_w = 5\%$, $\lambda = 0.034$, $t_{1/2} = 20$\,yr}

\addplot[black!70, densely dotted]
  {exp(-0.0565*x)};
\addlegendentry{$\tau_w = 10\%$, $\lambda = 0.057$, $t_{1/2} = 12$\,yr}

\draw[blue!80!black, thin, densely dotted] (axis cs:21,0) -- (axis cs:21,0.5);
\draw[blue!80!black, thin, densely dotted] (axis cs:0,0.5) -- (axis cs:21,0.5);
\node[font=\tiny, blue!80!black, fill=white, inner sep=1pt]
  at (axis cs:21,0.54) {$21$};

\draw[black!70, thin, densely dotted] (axis cs:12.3,0) -- (axis cs:12.3,0.5);
\node[font=\tiny, black!70, fill=white, inner sep=1pt]
  at (axis cs:12.3,0.54) {$12$};

\fill[orange!15] (axis cs:0,0) rectangle (axis cs:8,1.05);
\node[font=\scriptsize, black!60, rotate=90, anchor=south]
  at (axis cs:4,0.52) {policy cycle};

\end{axis}
\end{tikzpicture}
\caption{Relaxation toward the new steady-state wealth distribution
  after a tax change, for three tax rates.  Parameters:
  $\sigma = 0.30$, $\mu = 0.08$, $\delta = 1/30$ (generational
  turnover).  The vertical axis measures the distance
  $\|p(\cdot,t) - p_\infty\|$ normalised to unity at $t = 0$.
  At $\tau_w = 2\%$ the time-average growth rate remains positive
  ($v_\tau > 0$), so convergence depends entirely on demographic
  turnover ($\lambda = \delta$).  At $\tau_w = 5\%$, $v_\tau$ turns
  negative but only marginally, adding less than $4\%$ to the
  convergence rate.  Even at the extreme rate $\tau_w = 10\%$, the
  half-life is still over a decade.  The shaded region marks a
  typical electoral cycle ($\sim$4--8 years), during which the
  distribution has moved less than a quarter of the way to its new
  steady state in all three scenarios.}
\label{fig:relaxation}
\end{figure}

\section{Discussion and extensions}\label{sec:discussion}

\subsection{What the mapping enables}

The preceding sections established six results: neutrality is a
drift-shift symmetry (\Cref{sec:neutrality}) that is robust to
non-Gaussian returns and stochastic volatility
(\Cref{sec:robustness}), each distortion channel
breaks it in a classifiable way (\Cref{sec:channels}), the
Fokker--Planck equation governs the distributional consequences
(\Cref{sec:steady}), the Pareto tail steepness is a
drift-to-diffusion ratio whose sign is controlled by the time-average
growth rate (\Cref{sec:ergodicity}), and the spectral gap of the
Fokker--Planck operator governs how fast the distribution converges
to steady state (\Cref{sec:relaxation}).  Together, these open concrete
avenues that the standard finance analysis does not provide:
\begin{enumerate}
  \item \textbf{Empirical tests of symmetry breaking.}  Each distortion
    channel predicts a specific deformation of the propagator
    (\Cref{fig:channels}).  In principle, panel data on wealth
    trajectories before and after a tax change could be tested against
    these signatures---distinguishing, for example, broadening (liquidity)
    from truncation (migration).
  \item \textbf{Quantitative bounds on distortion magnitude.}  Because
    each channel modifies a specific coefficient in the Fokker--Planck
    equation, its effect can be bounded by estimating the relevant
    parameter (book-to-market dispersion, bid--ask spreads, migration
    elasticities).  The taxonomy converts a qualitative policy debate
    into a parameter-estimation problem.
  \item \textbf{Distributional dynamics.}  The Fokker--Planck equation
    describes not only the steady-state distribution but also the
    transient path toward it.  The spectral gap analysis of
    \Cref{sec:relaxation} shows that this path can be slow---decades
    for realistic parameters---implying that short-run distributional
    effects of a tax change may differ substantially from the long-run
    steady state.  These questions are inaccessible from within
    representative-agent asset pricing.
  \item \textbf{Policy timescales.}  Two results of the preceding
    analysis combine into a concrete policy message.  First, the
    ergodicity analysis (\Cref{sec:ergodicity}) shows that a
    stationary wealth distribution \emph{requires} the typical
    investor's log-wealth to shrink over time ($v_\tau < 0$).  In
    real economies, median wealth growth is indeed lower than mean
    wealth growth---this is not a pathology of the model but the
    standard mechanism by which multiplicative noise generates heavy
    tails.  The wealth tax steepens the Pareto tail by pushing the
    time-average growth rate further into negative territory, but the
    effect on the \emph{typical} investor is a uniform reduction in
    growth, not a reallocation across assets.  Second, the spectral
    gap (\Cref{sec:relaxation}) quantifies \emph{how long} these
    distributional shifts take.  For realistic calibrations
    (\Cref{fig:relaxation}), the half-life of the transition exceeds
    twenty years---far longer than a typical electoral or policy cycle.
    A government introducing a wealth tax to reduce inequality should
    therefore expect a slow, multi-generational adjustment.  Short-run
    revenue effects will materialise immediately (they depend only on
    the current distribution and the tax rate), but the distributional
    steady state lies decades away.  This mismatch between fiscal
    timescales and distributional timescales is invisible in static
    or representative-agent analyses; it emerges naturally from the
    Fokker--Planck framework.
\end{enumerate}

\subsection{Open questions}

Several extensions merit investigation.

\textbf{Detailed balance, market efficiency, and nonlinear
Fokker--Planck dynamics.}  The drift-shift symmetry of
\Cref{sec:neutrality} has a natural equilibrium interpretation.
Under the efficient market hypothesis, risk-adjusted returns are
unpredictable: the drift under the risk-neutral measure equals the
risk-free rate.  In Fokker--Planck language, this is a detailed-balance
condition---the probability current $J$ vanishes at steady state
\citep[cf.][Ch.~6]{LiviPoliti2017}.  The proportional wealth tax
preserves detailed balance because it shifts the drift uniformly,
which is another way to see why prices are unaffected.  The inelastic
markets hypothesis of \citet{GabaixKoijen2022} breaks detailed balance
by introducing flow-dependent pricing---a mechanism given a
microstructural interpretation by \citet{Bouchaud2021}---so that the
drift becomes a functional of the distribution itself.  Under inelastic markets, the wealth tax
creates a permanent probability current that has no counterpart in the
efficient-markets case, and this current drives a persistent price
impact.  Modelling this requires a nonlinear Fokker--Planck equation
of the form in \eqref{eq:ch5_fp}, where the drift is
self-consistently determined by the distribution.  This moves the
analysis from the linear, exactly solvable regime of the present paper
into the territory of nonlinear diffusion and mean-field models---a
substantial extension that we defer to future work.

\textbf{Heterogeneous growth regimes.}
The present paper assumes a common drift~$\mu$ and
volatility~$\sigma$ for all investors.  In practice, different
wealth levels may face systematically different growth rates---for
instance, if wealthy investors have access to higher-return asset
classes.  \citet{GabaixEtAl2016} show that introducing
wealth-dependent growth rates can dramatically accelerate the
convergence to the stationary distribution, resolving the ``slow
transition'' puzzle that arises when the spectral gap is dominated
by the demographic term~$\delta$ (\Cref{sec:relaxation}).
In the Fokker--Planck framework, this corresponds to a
state-dependent drift $v(x)$, which breaks the constant-coefficient
structure that underlies the exact results of this paper.  Whether the
neutrality result survives approximately in such a setting---and if
so, under what conditions---is an open question with direct empirical
relevance.

\textbf{Debt, leverage, and the net wealth tax base.}
The leverage amplification identified in \Cref{sec:ch1} implies that
the Fokker--Planck dynamics of \emph{taxable} net wealth can differ
substantially from those of market net wealth.  When assets are
assessed below market value but debt is fully deductible, the
effective drift and diffusion of the taxable distribution depend on
the leverage ratio, which is itself an endogenous choice variable.
A full treatment would model the joint dynamics of assets, debt,
and taxable net wealth as a multi-dimensional
Fokker--Planck system, with the leverage decision coupling the
dimensions.  This is particularly relevant for real estate, where
assessment discounts are large and mortgage financing is prevalent.
We defer the formal multi-dimensional treatment to a separate
paper on redistribution analysis \citep{Froeseth2026R}, which
provides a framework for incorporating it through the taxonomy of
Fokker--Planck modifications.

\textbf{Phase transitions.}
In principle, a sufficiently large change in the drift coefficient
could trigger a qualitative change in the distribution---for example,
a transition from a power-law tail to an exponential tail, or the
emergence of bimodality.  Whether such transitions occur at realistic
tax rates is an open empirical and theoretical question.

\appendix
\section{Derivation of the Pareto exponent}\label{app:pareto}

The formula for $\alpha(\tw)$ in \eqref{eq:pareto} does not follow from
pure GBM, which has no stationary distribution
(\Cref{sec:steady}).  It arises when multiplicative growth is
supplemented by an additive component---income, redistribution, or a
reflecting boundary---so that the process takes the Kesten form
\begin{equation}\label{eq:kesten}
  W_{n+1} = A_n\, W_n + B_n \,,
\end{equation}
where $A_n$ (the multiplicative factor from returns and tax) and $B_n$
(the additive income/redistribution term) are i.i.d.\ and independent of
each other.  The classical result of \citet{Kesten1973} states that if
$\E[\ln A] < 0$ (so that wealth does not diverge) and $B$ has
suitable integrability, the stationary distribution has a Pareto tail:
$\Pr(W > w) \sim w^{-\alpha}$, where $\alpha > 0$ is the unique
solution of
\begin{equation}\label{eq:kesten_condition}
  \E\bigl[A^\alpha\bigr] = 1 \,.
\end{equation}

In continuous time, the taxed GBM over an interval~$\dd t$ gives the
multiplicative factor
\[
  A = \exp\!\Bigl[\bigl(\mu - \tw - \tfrac{\sigma^2}{2}\bigr)\dd t
    + \sigma \sqrt{\dd t}\; Z\Bigr], \qquad Z \sim \mathcal{N}(0,1)\,.
\]
Evaluating the moment-generating function:
\[
  \E[A^\alpha]
  = \exp\!\Bigl[\alpha\bigl(\mu - \tw - \tfrac{\sigma^2}{2}\bigr)\dd t
    + \tfrac{\alpha^2 \sigma^2}{2}\dd t\Bigr] .
\]
Setting $\E[A^\alpha] = 1$ requires the exponent to vanish:
\[
  \alpha\bigl(\mu - \tw - \tfrac{\sigma^2}{2}\bigr)
    + \frac{\alpha^2 \sigma^2}{2} = 0 \,.
\]
Dividing by $\alpha \neq 0$ and solving:
\[
  \alpha = -\frac{2(\mu - \tw - \sigma^2/2)}{\sigma^2}
  = \frac{\sigma^2 - 2(\mu - \tw)}{\sigma^2}
  = 1 - \frac{2(\mu - \tw)}{\sigma^2} \,,
\]
which is \eqref{eq:pareto}.  The condition $\alpha > 0$ requires
$\mu - \tw < \sigma^2/2$, i.e.\ the log-wealth drift
$v_\tau = \mu - \tw - \sigma^2/2$ must be negative---the
multiplicative dynamics must be mean-reverting on average for a
stationary distribution to exist.

\section{Portfolio neutrality under stochastic volatility}
\label{app:stochvol}

This appendix provides the full derivation of
\Cref{prop:sv_neutral}, establishing that the optimal portfolio
weight under the Heston stochastic volatility model is independent
of the wealth tax rate.

\subsection*{Setup}

The investor maximises expected discounted CRRA utility
$\E\!\left[\int_0^T e^{-\delta t}\,
C_t^{1-\gamma}/(1-\gamma)\,\dd t\right]$
subject to the wealth dynamics~\eqref{eq:heston_dW} and the
variance dynamics~\eqref{eq:heston_v}.  The value function is
\[
  J(W, v, t) = \max_{C, w}\;
  \E_t\!\left[\int_t^T e^{-\delta(s-t)}
    \frac{C_s^{1-\gamma}}{1-\gamma}\,\dd s\right] .
\]

\subsection*{HJB equation}

By Bellman's principle, $J$ satisfies
\begin{equation}\label{eq:hjb_full}
\begin{split}
  0 = \max_{C, w} \biggl\{
    &\frac{C^{1-\gamma}}{1-\gamma} - \delta J + J_t \\
    &+ J_W\bigl[W(r_f + w(\mu - r_f) - \tw) - C\bigr] \\
    &+ J_v\,\lambda(\theta - v) \\
    &+ \tfrac{1}{2}J_{WW}\,w^2 v W^2
     + \tfrac{1}{2}J_{vv}\,\kappa^2 v \\
    &+ J_{Wv}\,wW\kappa v\rho
  \biggr\} .
\end{split}
\end{equation}

\subsection*{Separable value function}

The homogeneity of CRRA utility suggests the ansatz
\begin{equation}\label{eq:J_sep}
  J(W, v, t) = \frac{W^{1-\gamma}}{1-\gamma}\,f(v, t) \,,
\end{equation}
with $f > 0$.  The relevant partial derivatives are
\begin{alignat}{2}
  J_W &= W^{-\gamma} f \,, &\qquad
  J_{WW} &= -\gamma\, W^{-\gamma - 1} f \,, \notag\\
  J_v &= \frac{W^{1-\gamma}}{1-\gamma}\,f_v \,, &\qquad
  J_{vv} &= \frac{W^{1-\gamma}}{1-\gamma}\,f_{vv} \,, \notag\\
  J_{Wv} &= W^{-\gamma} f_v \,. && \label{eq:J_partials}
\end{alignat}

\subsection*{First-order condition for $w$}

Differentiating~\eqref{eq:hjb_full} with respect to~$w$ and
setting the result to zero:
\[
  J_W W(\mu - r_f) + J_{WW}\,w v W^2 + J_{Wv}\,W\kappa v\rho = 0 \,.
\]
Substituting~\eqref{eq:J_partials} and dividing by $W^{1-\gamma}f$:
\[
  (\mu - r_f) - \gamma\, w\, v + \frac{f_v}{f}\,\kappa v\rho = 0 \,.
\]
Solving for $w^*$:
\begin{equation}\label{eq:wstar_full}
  w^* = \frac{\mu - r_f}{\gamma v}
    + \frac{f_v}{f}\cdot\frac{\kappa\rho}{\gamma} \,.
\end{equation}
The tax rate~$\tw$ does not appear in~\eqref{eq:wstar_full}.

\subsection*{PDE for $f(v,t)$ and the Riccati argument}

Substituting the separable form~\eqref{eq:J_sep} and the optimal
controls into the HJB equation, and dividing by
$W^{1-\gamma}/(1-\gamma)$, yields a PDE for~$f$:
\begin{equation}\label{eq:f_pde_full}
\begin{split}
  0 &= f_t + \lambda(\theta - v)\,f_v
    + \tfrac{1}{2}\kappa^2 v\,f_{vv} \\
    &\quad + h(v;\,\gamma,\mu,r_f,\kappa,\rho)\,f \\
    &\quad + (1-\gamma)(r_f - \tw)\,f + g(f) \,,
\end{split}
\end{equation}
where $h(v;\,\cdot\,)$ collects the $v$-dependent terms arising from
the optimal portfolio (a quadratic in $f_v/f$) and $g(f)$ collects
the contribution from optimal consumption.  The key observation:
$\tw$ appears \emph{only} in the term $(1-\gamma)(r_f - \tw)f$,
which is independent of~$v$.

Using the exponential-affine ansatz
$f(v,t) = \exp(A(t) + B(t)\,v)$ \citep{ChackoViceira2005},
substitution into~\eqref{eq:f_pde_full} separates into:
\begin{itemize}
  \item A \textbf{Riccati equation for $B(t)$}, arising from the
    terms proportional to~$v$:
    \[
      \dot{B} = -\lambda B + \tfrac{1}{2}\kappa^2 B^2
        + h_1(\gamma, \mu, r_f, \kappa, \rho) \,,
    \]
    where $h_1$ absorbs the $v$-coefficient from~$h$.  This equation
    involves only the parameters of the variance dynamics and the
    risk-return trade-off---\emph{not}~$\tw$.
  \item An \textbf{ODE for $A(t)$}, arising from the constant terms:
    \[
      \dot{A} = \lambda\theta\,B
        + (1-\gamma)(r_f - \tw) + h_0 \,,
    \]
    which absorbs~$\tw$ alongside other constants.
\end{itemize}
Since $f_v/f = B(t)$ and $B(t)$ is independent of~$\tw$, the hedging
demand $\frac{f_v}{f}\cdot\frac{\kappa\rho}{\gamma} = B(t) \cdot
\frac{\kappa\rho}{\gamma}$ is tax-invariant.  Combined with the
tax-invariance of the myopic demand, this establishes that $w^*$ is
independent of~$\tw$.  \qed

\section{Geometric interpretation of neutrality}\label{app:geometry}

The drift-shift symmetry and the CRRA separability result admit a
geometric formulation that unifies the algebraic arguments of
\Cref{sec:robustness} and connects them to structures familiar from
theoretical physics.  The formulation uses the language of fiber
bundles and connections \citep[see][Ch.~9--10, for a physicist's
introduction]{Nakahara2003}; readers interested in the application
of gauge-theoretic ideas to finance may consult
\citet{Ilinski2001}.

\subsection*{The state space as a fiber bundle}

Consider the investor's full state as a point in a product space: the
\emph{base} is the log-wealth coordinate $x \in \mathbb{R}$ (or, under
stochastic volatility, the pair $(x, v) \in \mathbb{R} \times
\mathbb{R}_+$), and the \emph{fiber} above each base point is the
portfolio simplex $\Delta^{N-1} = \{\mathbf{w} \in \mathbb{R}^N :
\sum_i w_i = 1\}$.  The investor's optimisation problem selects a
\emph{section} of this bundle: a rule $\mathbf{w}^*(x, \mathbf{X}, t)$
that assigns an optimal portfolio to each state.

Under CRRA utility, the value function separates as
$J = W^{1-\gamma} f(\mathbf{X}, t)/(1-\gamma)$, and the first-order
conditions yield a portfolio $\mathbf{w}^*$ that is independent of the
base coordinate~$x$ (equivalently, of~$W$).  In differential geometry
language, the optimal section is \emph{horizontal}: it does not vary
along the base.  The connection on the bundle is \emph{flat}---there is
no curvature, no holonomy, and no wealth-dependence in the portfolio
prescription.

The proportional wealth tax acts as a \emph{vertical automorphism}: it
translates every point along the base ($x \to x - \tw t$) without
rotating the fiber.  Because the connection is flat and the tax acts
only vertically, the horizontal section is invariant.  This is the
geometric content of portfolio neutrality.
\Cref{fig:fiber_bundle} illustrates the construction.

\begin{figure}[H]
\centering
\begin{tikzpicture}[
  >=Stealth,
  fiber/.style={draw=black!30, thin},
  section/.style={very thick},
  taxarrow/.style={->, thick, red!70!black},
  scale=0.92, every node/.style={scale=0.92},
]

\node[font=\bfseries\small, anchor=south] at (3.2, 5.8) {(a) Neutrality: flat connection};

\draw[->, thick] (-0.3, 0) -- (6.8, 0)
  node[below right, font=\small] {$x = \ln W$};

\draw[->, thick] (-0.3, 0) -- (-0.3, 5.5)
  node[above right, font=\small] {$w^*$};

\foreach \xpos in {0.5, 1.5, 2.5, 3.5, 4.5, 5.5} {
  \draw[fiber] (\xpos, 0.1) -- (\xpos, 5.2);
}

\draw[section, blue!70!black] (0.3, 3.2) -- (6.2, 3.2);
\node[font=\small, blue!70!black, anchor=west] at (6.3, 3.2) {$\mathbf{w}^*$};

\foreach \xpos in {0.5, 1.5, 2.5, 3.5, 4.5, 5.5} {
  \fill[blue!70!black] (\xpos, 3.2) circle (2pt);
}

\draw[taxarrow] (5.2, 0.7) -- (3.8, 0.7)
  node[midway, above, font=\small] {$\tw t$};
\draw[taxarrow] (5.2, 1.3) -- (3.8, 1.3);

\node[font=\footnotesize, text=black!60, anchor=south west] at (0.6, 4.8)
  {fibers $\Delta^{N-1}$};

\draw[decorate, decoration={brace, amplitude=4pt, mirror}]
  (0.3, -0.4) -- (6.2, -0.4)
  node[midway, below=5pt, font=\footnotesize] {base: log-wealth};

\begin{scope}[xshift=8.2cm]

\node[font=\bfseries\small, anchor=south] at (3.2, 5.8)
  {(b) Distortion: curved connection};

\draw[->, thick] (-0.3, 0) -- (6.8, 0)
  node[below right, font=\small] {$x = \ln W$};

\draw[->, thick] (-0.3, 0) -- (-0.3, 5.5)
  node[above right, font=\small] {$w^*$};

\foreach \xpos in {0.5, 1.5, 2.5, 3.5, 4.5, 5.5} {
  \draw[fiber] (\xpos, 0.1) -- (\xpos, 5.2);
}

\draw[section, red!60!black]
  plot[smooth, domain=0.3:6.2, samples=50]
    (\x, {2.4 + 0.35*sin(60*\x) + 0.12*\x});
\node[font=\small, red!60!black, anchor=west] at (6.3, {2.4 + 0.35*sin(60*6.2) + 0.12*6.2}) {$\mathbf{w}^*(x)$};

\foreach \xpos in {0.5, 1.5, 2.5, 3.5, 4.5, 5.5} {
  \pgfmathsetmacro{\yval}{2.4 + 0.35*sin(60*\xpos) + 0.12*\xpos}
  \fill[red!60!black] (\xpos, \yval) circle (2pt);
}

\draw[section, blue!70!black, dashed, opacity=0.3] (0.3, 3.2) -- (6.2, 3.2);

\draw[<->, thin, black!60]
  (4.5, 3.2) -- (4.5, {2.4 + 0.35*sin(60*4.5) + 0.12*4.5});
\node[font=\footnotesize, text=black!60, anchor=west] at (4.6, 2.9)
  {distortion};

\draw[decorate, decoration={brace, amplitude=4pt, mirror}]
  (0.3, -0.4) -- (6.2, -0.4)
  node[midway, below=5pt, font=\footnotesize] {base: log-wealth};

\end{scope}

\end{tikzpicture}
\caption{Geometric interpretation of neutrality and its breakdown.
  \textbf{(a)}~Under CRRA preferences, the optimal portfolio
  $\mathbf{w}^*$ is a horizontal section of the fiber bundle: it does
  not vary with log-wealth~$x$.  The proportional wealth tax (red
  arrows) translates the base leftward without rotating the fibers,
  leaving the section invariant.
  \textbf{(b)}~When a distortion channel is active (e.g.\ book-value
  assessment, liquidity frictions), the section becomes
  wealth-dependent: the connection acquires curvature, and the
  portfolio varies with~$x$.  The dashed line shows the flat
  (neutral) section for comparison.}
\label{fig:fiber_bundle}
\end{figure}
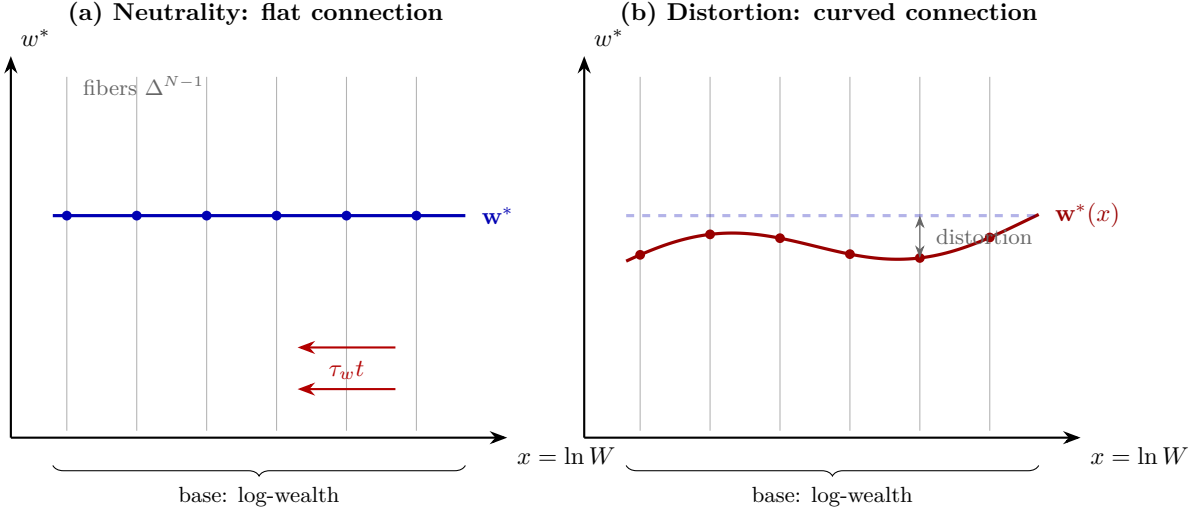

\subsection*{Galilean symmetry in log-wealth space}

The drift-shift transformation
$\mathcal{T}_\tau: v \mapsto v - \tw$, $D \mapsto D$
(\Cref{sec:neutrality}) is the analogue of a \emph{Galilean boost} in
classical mechanics
\citep[for the symmetry group of the diffusion equation,
see][Ch.~6]{Risken1989}.
In the Galilean group, a uniform velocity shift
preserves all relative velocities and all forces (which depend on
position differences).  Here, the tax shifts all drift velocities by
the same constant~$\tw$, preserving all relative drifts, diffusion
coefficients, and Sharpe ratios.

The portfolio weights, being functions of excess returns
$\mu_i - r_f$ and the covariance matrix~$\bm{\Sigma}$, are
\emph{Galilean invariants}: quantities that depend on velocity
differences, not absolute velocities.  The risk-free rate $r_f$
transforms under the same boost ($r_f \to r_f - \tw$), so the excess
return $\mu_i - r_f$ is invariant---exactly as relative velocities are
invariant under a Galilean transformation.

This analogy makes precise the paper's characterisation of the tax as a
``uniform external field'' (\Cref{sec:neutrality}): a spatially
uniform field shifts the equilibrium of every degree of freedom by the
same amount, preserving the relative structure.

\subsection*{Noether's theorem and the conserved charge}

The one-parameter family of transformations
$\{\mathcal{T}_\tau\}_{\tau \geq 0}$ is a continuous symmetry group of
the portfolio optimality conditions.  Noether's theorem associates a
conserved quantity to each continuous symmetry.  The conserved
``charge'' here is the excess return vector
\begin{equation}\label{eq:conserved_charge}
  \mathbf{q} \;\equiv\; \boldsymbol{\mu} - r_f\,\mathbf{1} \,,
\end{equation}
which is invariant under $\mathcal{T}_\tau$ because
$(\boldsymbol{\mu} - \tw\,\mathbf{1}) - (r_f - \tw)\,\mathbf{1}
= \boldsymbol{\mu} - r_f\,\mathbf{1}$.  The optimal portfolio depends
only on this conserved charge and the covariance matrix (which lives
entirely in the diffusion sector and is therefore $\tw$-independent).
Portfolio neutrality is the statement that the optimal portfolio is a
function of conserved quantities only.

\subsection*{Symmetry breaking as curvature}

Each distortion channel of \Cref{sec:channels} has a geometric
reading: it introduces \emph{curvature} into the previously flat
connection, making the optimal portfolio wealth-dependent.
Book-value assessment creates a twist between the base and the
fiber (the effective tax rate depends on unrealised gains, which
depend on~$x$).  Liquidity frictions make the connection
density-dependent.  Migration introduces a boundary that breaks
the translation symmetry of the base space.  In each case, the
horizontal section ceases to exist globally: the portfolio
must now vary with the base coordinate, and neutrality fails.

\subsection*{The two mechanisms revisited}

The fiber bundle picture clarifies why the location-scale and
stochastic volatility generalisations (\Cref{sec:robustness}) rest on
different structures.  The location-scale result is a statement about
the \emph{base}: the translation symmetry $x \to x - \tw t$ holds
regardless of the noise distribution, because it is a property of the
tax mechanism (a deterministic multiplicative drain).  No reference to
the fiber or to preferences is needed.

The stochastic volatility result is a statement about the
\emph{connection}: under CRRA, the connection remains flat even when the
base is enlarged from $x$ to $(x, v)$.  The new coordinate~$v$ adds
dimensionality to the state space and introduces hedging demand into the
portfolio, but the CRRA homogeneity ensures that $\mathbf{w}^*$ remains
independent of~$x$.  The tax, acting only along~$x$, leaves the
(still flat) connection invariant.

\subsection*{Acknowledgements}
The author acknowledges the use of Claude (Anthropic) for assistance with
literature review, \LaTeX{} typesetting, mathematical exposition, and
editorial refinement, and Lemma (Axiomatic AI) for review and proof
checking. All substantive arguments, economic reasoning, and conclusions
are the author's own.

\bibliographystyle{plainnat}

\end{document}